\documentclass[10pt]{article}
\usepackage{geometry}
\usepackage{setspace}
\usepackage{algorithm}
\usepackage{algorithmic}
\usepackage{url} 
\usepackage{cite}
\usepackage{amsmath,amssymb,amsfonts}
\usepackage{graphicx}
\usepackage{textcomp}
\usepackage{xcolor}
\usepackage{subfigure}
\usepackage{bm}
\usepackage{appendix}
\usepackage{amsthm}
\usepackage{setspace}  

\theoremstyle{definition}
\newtheorem{theorem}{\textbf{Theorem}}
\newtheorem{definition}{\textbf{Definition}}

\newtheorem{remark}{\textbf{Remark}}

\newtheorem{lemma}{\textbf{Lemma}}
\newtheorem{corollary}{\textbf{Corollary}}

\allowdisplaybreaks[4]
\def\BibTeX{{\rm B\kern-.05em{\sc i\kern-.025em b}\kern-.08em
    T\kern-.1667em\lower.7ex\hbox{E}\kern-.125emX}}
\newcommand{\todo}[1]{\textcolor{black}{#1}}

\title{\todo{A Prudent Framework for Understanding Risk-Awareness \\in Demand Response}}
\author{
    Liudong Chen\textsuperscript{}, Bolun Xu \\
    Earth and Environmental Engineering, Columbia University \\
    \{lc3671, bx2177\}@columbia.edu  \\
    500 W 120th Street, NYC, NY, 10027, USA \\
}
\date{} 

\begin{document}

\maketitle

\begin{abstract}
\todo{We show that risk-aware behaviors in demand response originate from superquadratic state-dependent cost functions and price uncertainty with skewed distributions. We obtain such results through developing a novel theoretical demand response framework that combines non-anticipatory multi-stage decision-making with superquadratic cost functions. We introduce the concept of prudent demand, defined by a positive third-order derivative of the cost function, which is the first principle for risk-averse behavior despite a risk-neutral objective. Our analysis establishes that future price uncertainty affects immediate consumption decisions, and the extent of this response scales proportionally with the skewness of the price distribution. We visualize our theoretical findings through numerical simulations and illustrate their practical implications using a real-world case study.}
\end{abstract}


\section{Introduction}
Utility companies and load-serving entities are increasingly introducing dynamic pricing plans to promote flexible demand solutions, including storage units, smart appliances, and electric vehicles~\cite{intro_1}, with the expectation that consumers will strategically adjust their electricity usage to minimize costs.
One prominent form of dynamic pricing is real-time pricing, in which electricity prices fluctuate over time and are disclosed only at the point of delivery~\cite{optimal_response}. This price variability, influenced by wholesale market rates or local grid conditions, enables incentive-based demand response (DR)~\cite{doe_document} at relatively low cost, thereby improving system reliability and reducing demand during system contingencies~\cite{DR_qiadratic_reason}.
\todo{For instance, PJM Interconnection (PJM)’s economic DR program partially exposes participating customers to real-time wholesale price variability, encouraging strategic demand adjustments~\cite{PJM}. Similarly, Electric Reliability Council of Texas (ERCOT)’s real-time market exposes wholesale customers to highly variable prices updated every five minutes, while its four coincident peaks (4CP) program imposes substantial charges based on customers' demand during system peak hours identified retrospectively~\cite{4CP_ercot}.}

While real-time tariffs provide utilities with the flexibility to adjust prices in response to current grid conditions, they also introduce significant uncertainty for consumers. Exposure to fluctuating electricity prices elicits complex risk-aware behaviors~\cite{robust}. \todo{As suggested by numerous industry solutions, customers are advised to schedule curtailment plans or adjust flexibility resources proactively based on anticipated wholesale price patterns~\cite{PJM_CP}. To mitigate the charge associated with the ERCOT 4CP program, both industry practices and academic research recommend that customers proactively adjust their consumption based on peak load predictions~\cite{baosen_cdc, 4CP}. Similarly, when significant events are forecasted, such as during the 2024 solar eclipse, many system operators schedule backup resources (e.g., batteries, gas turbines) in advance to mitigate the risks associated with extreme events~\cite{NREL}. On the residential side, programs such as Tesla’s Storm Watch automatically charge home batteries to full capacity in anticipation of severe weather forecasts~\cite{tesla}, reflecting consumers’ willingness to act in advance of low-probability high-impact events. Collectively, these examples indicate strong behavior: specifically, risk-averse behaviors that prompt early action and skewness-averse behaviors that elicit stronger responses to extreme tail events. }

\todo{Real-world experiences highlight the need to systematically study consumer risk-aware behaviors to better support program and tariff design. A natural approach is to conduct real-world controlled experiments by offering dynamic prices, and perform subsequent analysis on the observed consumption outcomes~\cite{learning_model_free}. However, such experiments are often impractical due to infrastructure constraints, privacy concerns, and regulatory hurdles. More fundamentally, electricity is widely regarded as a basic human necessity, and both the public and utilities are strongly averse to exposing consumers, especially vulnerable populations, to price volatility purely for experimental purposes~\cite{ansarin2022review,horowitz2014equity}. These ethical and political sensitivities make it difficult to justify controlled price manipulation at scale. As a result, empirical studies remain limited, and it is often more practical to use decision-theoretic models that capture consumer costs, constraints, and uncertainty within an optimization framework.}

\todo{Previous literature has commonly employed quadratic or piecewise-linear utility function formulations to approximate decision behavior under stochastic settings, although in the deterministic setting, there are higher-order analyses. However, these approaches generally fail to capture the true operational characteristics of household appliances, thermal comfort preferences, and battery storage capabilities~\cite{optimal_response_GM, Stackelberg_game, battery}. For example, thermal discomfort does not decrease linearly with temperature adjustments~\cite{thermal_nonlinear}, nor do energy usage patterns evolve purely in a quadratic manner; rather, they are influenced by high-dimensional, non-linear individual preferences~\cite{feedback}. Similarly, batteries and many household appliances operate under strict physical constraints (e.g., hard capacity limits, ramp rate restrictions) that are inadequately represented by quadratic models, which inherently apply only soft penalties. 
Moreover, most stochastic programming research considers uncertainty only through expected values~\cite{powell2019unified}, often neglecting distributional effects that naturally arise in realistic, skewed distributions~\cite{schuhmacher2021justifying}. Although some studies introduce risk aversion by employing mean-variance or conditional value-at-risk (CVaR) formulations~\cite{oum2006hedging}, these approaches generally rely on strong distributional assumptions and presuppose risk-averse preferences without explaining their origins or deriving them naturally from the structural properties of the decision-making model. Even in studies that consider higher-order formulations within stochastic settings, the emphasis is typically placed on developing algorithms for simulation, rather than on theoretical exploration to understand the risk-aware behaviors~\cite{nalpas2017portfolio}.}

\todo{This gap naturally underscores the need for a more sophisticated utility function formulation that can capture these features and more accurately reflect real-world risk-aware decision behavior under skewed uncertainty conditions. In light of this, a critical research question arises: \emph{What is the first principal of demand's risk-aware decision behavior under price uncertainty with skewed distributions}. Addressing this question is the main focus of this paper.}



\subsection{\todo{Summary of contributions and implications}}
\todo{In this work, we establish a theoretical framework to model demand behavior under volatile future electricity prices. We focus on a generic demand model---representative of flexible loads such as battery storage or thermostatically controlled Heating, Ventilation, and Air Conditioning (HVAC) systems---that accommodates different cost function structures, specifically, quadratic and superquadratic forms. The latter, characterized by nonzero third-order derivatives, enables us to uncover the structural foundations of risk-aware behavior in demand response. While prior work has extensively examined the role of expectation in uncertainty, it is increasingly evident that higher-order distributional features, particularly skewness, play a critical role in shaping decision-making. We therefore focus on the more realistic skewed distribution and capture the distributional effect on the decision behavior. 
To this end, we introduce the concept of \emph{prudent demand}, motivated by the classical economic notion of prudence, which describes the sensitivity of optimal responses to risk and captures higher-order moment risk behavior~\cite{prudent_earliet, ebert2024first, risk_averse}. 
 }

 We contribute to the literature by:
\begin{itemize}
    \item \textbf{Demonstrating the limitations of quadratic demand models}: We prove that these models are distribution-insensitive and thus fail to capture meaningful behavioral responses to skewed price uncertainty or shifts in the price distribution.
    
    \item \textbf{Establishing that superquadratic cost functions naturally induce prudent demand}: Even under risk-neutral objectives, such functions result in precautionary behavior, where future price distribution influences immediate consumption decisions. This connects operational models to classical notions of prudence in economics and reveals that commonly used risk-averse formulations can be viewed as approximations of deeper structural properties.
    
    \item \textbf{Deriving skewness-averse behavior analytically}: We show that the magnitude of demand adjustment scales with the skewness of price distribution, offering new insights into how rare, high-impact events influence electricity consumption decisions.
\end{itemize}

Our findings carry several practical implications:
\begin{itemize}
    \item \textbf{For system operators}: The widespread use of quadratic demand models may result in misaligned incentives or poor forecasts. Incorporating behaviorally aware formulations can utilize the 'built-in' flexibility of demand and improve demand-side management under uncertainty.
    
    \item \textbf{For policymakers}: Understanding that consumers respond to distributional risk supports strategies that include pre-event communication and better-designed dynamic pricing to enhance grid reliability.
    
    \item \textbf{For researchers}: Our results provide a principled method for embedding higher-order distributional effects into demand modeling, offering a bridge between economic theory and real-world energy system operations.
\end{itemize}

The remaining of the paper is organized as follows: Section II reviews related literature, \todo{Section III introduces the model and preliminaries, Section IV introduces the distribution-insensitive demands' definition and conditions, Section V extends to prudent demands' analysis in terms of definition, conditions, and revelation, Section VI describes simulation results under DR settings, and Section VII concludes the paper and discuss practical implications.}

\section{Background and literature review}
\subsection{Dynamic pricing for electricity consumers}
Dynamic pricing schemes, also known as time-varying tariffs, are designed to leverage demand-side flexibility but also introduce greater variability into electricity prices~\cite{tariff_dynamic1}.
\todo{Electricity prices in wholesale markets are inherently volatile~\cite{TX_realtime}, and in regions such as PJM and ERCOT, large industrial or commercial consumers can elect to receive real-time wholesale prices directly~\cite{PJM_real}.} During the 2021 Winter Storm Uri, wholesale real-time prices surged to \$9,000/MWh in Texas, resulting in some consumers receiving electricity bills exceeding ten thousand dollars~\cite{nyt}. \todo{In addition, programs such as ERCOT’s 4CP impose charges based on peak periods that are identified ex post to the following year’s bills. The resulting charges introduce significant uncertainty for consumers~\cite{4CP_ercot}.}

On the retail side, utilities are increasingly adopting dynamic tariffs to encourage demand response. Many utilities offer time-of-use (ToU) tariffs for residential customers, with different rates for peak and off-peak periods, aiming to provide consumers with opportunities for savings. Examples include Con Edison in New York City~\cite{Con_ed_ny}, PG\&E in California~\cite{PGE}, and the Salt River Project (SRP) in Arizona~\cite{SRP}. Beyond ToU, utilities are now experimenting with more advanced time-varying tariffs. For instance, OG\&E in Oklahoma implements a smart hours pricing scheme that varies prices within the designated ToU peak window~\cite{OGE}, while Ameren in Illinois offers a Power Smart Pricing program that issues hourly electricity prices to residential customers on a day-ahead basis~\cite{real_time_price}. Although many of these programs issue prices in advance and thus limit real-time uncertainty for consumers, they increasingly link retail prices to wholesale market conditions, indicating a broader trend toward dynamic pricing models that aim to enhance consumer flexibility in electricity consumption.

\subsection{\todo{Responsiveness of electricity consumers}}
\todo{The key to effective dynamic tariff design is understanding how consumers respond to price signals. However, privacy, affordability, and fairness concerns make it difficult to conduct controlled real-world experiments, particularly under variable real-time tariff structures.} As a result, much of the prior research has focused on model-based approaches, employing speculative utility functions to represent consumer decision-making processes~\cite{DR_intro_earlist, pennings2003shape}. These models typically aim to capture behaviors such as appliance usage shifting or thermal comfort adjustments in response to dynamic prices~\cite{optimal_response_GM, learning_parameter}.

\todo{The optimal solution to a consumer's utility maximization problem comes from the first-order optimality condition, which requires a concave and monotonic increasing function~\cite{optimal_response, DR_qiadratic_reason}. Quadratic functions are among the simplest forms that satisfy these conditions while also maintaining computational tractability and allowing for closed-form solutions.} However, the reliance on quadratic formulations represents a trade-off for mathematical convenience~\cite{DR_qiadratic_review}, which may not accurately capture real-world consumer response behaviors. In parallel, learning-based methods have been developed to infer consumer consumption patterns by estimating utility function parameters from observed data~\cite{learning_data, learning_parameter}, integrating user feedback into the control loop~\cite{feedback}, and even modeling individual response behaviors through model-free formulations~\cite{learning_model_free}.
Nevertheless, these methods face the inherent risk of extrapolating learned behavior to unseen price patterns.

\todo{Electricity consumers' responses to uncertain prices are widely studied, often considering uncertainty only through expected values~\cite{powell2019unified}. To capture risk-aware behavior, some works intentionally introduce risk aversion considerations, such as CVaR~\cite{cvar2}, to expect consumers to demonstrate risk-averse behaviors. Other studies apply mean-variance approaches to model price and quantity risk in forward contracts~\cite{oum2006hedging}, or consider the variance of consumers' random responsiveness when designing demand response contracts~\cite{aid2022optimal}. Moving toward more realistic settings, some research treats consumers' price responses as a decision-dependent source of uncertainty, influencing system operations such as unit commitment scheduling~\cite{lejeune2024profit}. However, existing approaches generally presume the form of risk aversion exogenously, either through ad hoc risk measures or utility assumptions, without explaining how risk-averse behaviors arise structurally from the decision model itself. Moreover, prior work has largely focused on modeling techniques and empirical estimation, rather than offering a theoretical analysis of the fundamental structure underlying demand risk-aware decision behavior. }

\subsection{\todo{Prudence and higher moment analysis}}
\todo{Prudence is a classical economic concept that captures the sensitivity of optimal decision-making to the presence of future risk~\cite{prudent_earliet}. While related to the notion of risk aversion first introduced by Pratt~\cite{pratt_risk}, prudence, introduced by Kimball from a precautionary saving perspective~\cite{prudent_earliet}, characterizes the agent’s willingness to adjust current decisions, such as saving more today, to prepare for future uncertainty. Formally, risk aversion is a lower-order risk attitude characterized by the second-order derivative of the utility function, whereas prudence corresponds to a higher-order risk attitude, associated with the third-order derivative~\cite{risk_averse}. As the order of risk attitudes increases, the model can capture the influence of higher-order moments of the uncertainty distribution, including variance, skewness, and kurtosis.
Several works have focused on mean-variance-skewness-kurtosis analysis, theoretically incorporating higher moments into risk-based optimization frameworks to reflect richer behavioral preferences and to distinguish risk-averse and risk-loving tendencies~\cite{zhang2020supply, ren2024risk}. Numerical algorithms have also been developed to solve these multi-moment optimization problems~\cite{nalpas2017portfolio}.}

\todo{As prudence has been shown to imply inherent risk-averse behavior through precautionary saving, it naturally leads to loss-averse decision-making when it comes to action~\cite{loss_averse}. From a theoretical perspective, prudent agents are proven to exert greater effort in earlier periods to accumulate larger wealth reserves for future risk periods~\cite{two_period_1}, even under deeper uncertainty, such as ambiguity, where the risk distribution is unknown~\cite{ambiguity}. Prudence has also been linked to skewness preference: agents are generally willing to bear a higher level of overall risk or accept lower expected payoffs when the risk is right-skewed, which naturally motivates left-skewness-seeking behavior~\cite{skewness_left}. This behavior aligns with the concept of prudence, wherein prudent individuals prefer to apportion risk asymmetrically to mitigate exposure to adverse outcomes~\cite{skewness_1}. Empirical evidence further supports these theoretical findings, demonstrating observed correlations between prudence and skewness preferences in experimental studies~\cite{experiment_2}.}

All of these behaviors are ultimately reflected in the optimal decision-making process under uncertainty. When focusing only on the time when risk materializes, prudence, through its higher-order properties, captures more nuanced behaviors than basic risk aversion by influencing discount factors across different outcome timings~\cite{time}. Prudence also implies that knowing less necessitates doing more~\cite{technology_risk}, leading individuals to become more accepting of the costs associated with risk management methods in the face of an unforeseen future~\cite{change_more}. This is because prudent individuals can compensate for uncertainty and lower their exposure to risk by preparing in advance, a behavior structurally reflected as precautionary effort~\cite{high_order}. \todo{Although many studies have examined optimal decision-making under uncertainty, the application and influence of prudence on electricity demand behavior has not yet been systematically explored, particularly within a sequential decision-making framework that incorporates time dependency, state transitions, and a dynamic cost structure.}

\section{Model and Preliminaries}
In this section, we formulate our model and introduce the definitions. We consider demand with a \emph{risk-neutral} cost-saving objective and linear system model responding to future electricity price uncertainty in a discrete time-varying system with stage $t\in [1,T]$. \todo{We denote the uncertain price at stage $t$ as $\lambda_t$, which comes from a distribution $\Lambda_{t}$, $\Lambda_{t} \in \mathcal{L}_{t}$. The power action and states at stage $t$ are denoted by $p_t, x_t$, respectively, where the state value can represent factors such as battery state-of-charge (SOC) or temperature. We use $C_t(x_t), G_t(p_t)$ to represent the state and action penalty cost functions, respectively. Here, we use the cost function to model the bounds on decision variables $x_t$ and $p_t$ (See Remark \ref{constraint_remark}).} The objective is to minimize the expectation of all stages' costs,
\begin{subequations} \label{original_problem}
\begin{align}
    \min_{p_t}&\;\mathbb{E}_{\Lambda_t} \sum_{t=1}^T \Big[\lambda_t p_t + C_t(x_{t}) + G_t(p_t)\Big] + V_T(x_T),\\
    \text{s.t. } & x_t  = A x_{t-1} + p_t, \label{cons} \\
    & \text{$p_t$ is non-anticipatory} \label{cons1}
 \end{align} 
\end{subequations}
where $V_T(x_T)$ is the end state-value function, representing the state value at the final stage for value continuity, and can be set to zero to show no final value. \eqref{cons} is the state transition constraint with the state discount factor $A\leq 1$. Finally, \eqref{cons1} states that the control is non-anticipatory to reflect the nature of multi-stage stochastic decision-making~\cite{shapiro2021lectures}.

\begin{definition} \label{def1} \emph{Normalized power and state cost.} We apply unit normalization and assume the system is in equilibrium at zero power and state to simplify the model and highlight our focus on disturbances and variations. For example, a HVAC system is set to control the room temperature to 25 Celsius, and then in our normalized system, the room temperature corresponds to $x_t = 0$ and $C_t(0) = 0$ represents no thermal discomfort. Formally, the normalization provides $C_t(0) = 0$ and $G_t(0) = 0$, with $C_t(x_t) \geq 0$ and $G_t(p_t) \geq 0$. We also define $G_t$ and $C_t$ as continuous and convex. Fig.~\ref{C,G,F} provides an example of the power and state cost functions.
\end{definition}

\begin{remark} \emph{\todo{Modeling hard and soft constraints.}}\label{constraint_remark}
\todo{We embed the state cost function $C_t$ and action cost function $G_t$ into the objective as penalty terms to represent the soft and hard bounds on states and actions. Lower-order functions, such as quadratics, impose moderate penalties and allow for controlled constraint violations, effectively modeling soft bounds. In contrast, higher-order functions, such as log-barrier or exponential functions, impose steep penalties as the state or action variables approach their limits, thus functionally equivalent to hard constraints.}
\end{remark}
\begin{figure}
    \centerline{\includegraphics[width=0.48\textwidth]{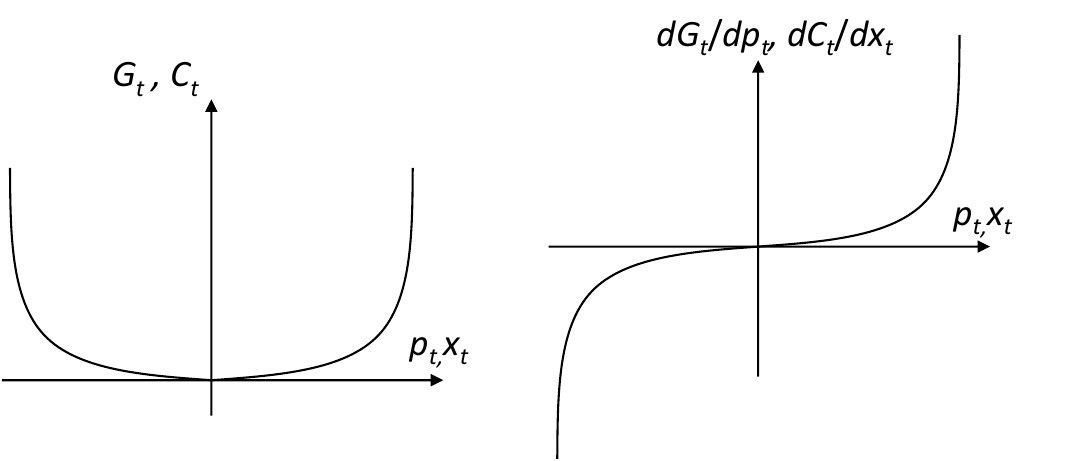}}
    \caption{Graph of function $C_t, G_t$ and their derivative.}
    \label{C,G,F}
\end{figure}

\begin{definition} \label{def2} \emph{Quadratic and super quadratic function.}
    We denote $a_{\mathrm{p}}$ as the function parameter and define the quadratic function as follows 
    \begin{subequations}
    \begin{align}
    G_t(p_t) = \frac{a_{\mathrm{p}}p_t^2}{2}, a_{\mathrm{p}} > 0,
\end{align}

    We define a superquadratic function as a function whose third-order derivative is positive for positive arguments,
    \begin{align}
        \frac{\partial^3 G_t(p_t)}{\partial p_t^3} >0 \text{ for } p_t>0.
    \end{align}          
    \end{subequations}
\end{definition}

\textbf{Stochastic dynamic programming reformulation.} To reflect the non-anticipative nature in the sequential decision-making under uncertainty, we use stochastic dynamic programming to reformulate \eqref{original_problem} by working backward and recursively solving a single-stage optimization, i.e., $\forall t\in [1,T]$: 
\begin{subequations}\label{eq1}
\begin{align}
    Q_{t-1} (x_{t-1}|\lambda_t) = &\min_{p_t} \lambda_t p_t + C_t (x_t) + G_t(p_t) + V_t(x_t)\\
    V_t(x_t) =& \mathbb{E}_{\Lambda_{t+1}}[Q_t(x_t|\lambda_{t+1})] \label{3b}\\
    \text{s.t.}\ x_t =& A x_{t-1} + p_t.  
\end{align}
\end{subequations}
where $Q_{t-1} (x_{t-1}|\lambda_{t})$ is the action-value function at stage $t-1$ parameterized by the price $\lambda_t$ at stage $t$, and $V_t(x_t)$ is the state-value function at stage $t$.

We show the notation used in the whole paper. We use $g_t(p_t) = \partial G_t(p_t)/ \partial p_t$, 
$c_t(x_t) = \partial C_t(x_t)/ \partial x_t$,
$v_t(x_t) = \partial V_t(x_t)/ \partial x_t$,
$q_t(x_{t}|\lambda_{t+1}) = \partial Q_t(x_t|\lambda_{t+1})/ \partial x_t$ to express the first-order derivative of the action cost function, state cost function, state-value function, and action-value function, respectively, and above $\cdot$ to express second-order derivative, such as $\dot{g}$. We also combine the derivative of state-related cost and express it as $h_t(x_t)$, i.e., $h_t(x_t)=c_t(x_t)+v_t(x_t)$. 

\section{Distribution-insensitive demand models}
In this section, we show the conditions for a demand model described in (\ref{eq1}) to be distribution-insensitive: the demand only responds to changes in future price expectations but is insensitive to changes in the price distribution if the expectations remain the same. The following theorem formally introduces this characteristic. 
\begin{theorem}\emph{Distribution-insensitive demand models.}\label{price_invariant_theorem}
    Consider the price $\lambda_{t}$ come from distribution $\Lambda_{t}, \Gamma_{t}$, $\{\Lambda_{t},\Gamma_{t}\} \in \mathcal{L}_{t}$ at stage $t$, \todo{where $\mathcal{L}_t$ denotes the set of distribution} with a fixed expectation $\mathbb{E}_{\Lambda_{t}}[\lambda_{t}] = \mathbb{E}_{\Gamma_{t}}[\lambda_{t}]$.
    Then a demand model described in (\ref{eq1}) with quadratic state cost function $C_t$ and action cost functions $G_t$ following Definitions \ref{def1} and \ref{def2} satisfies the \emph{distribution-insensitive} conditions \todo{that consumption at stage $t$ 
    is independent of all future distributions with fixed expectations at stage $\tau,\tau\in [t+1,T]$, i.e.,}
    \begin{align}
        \todo{\mathbb{E}_{\Gamma_{\tau}} [Q_{t}(x_{t}|\lambda_{\tau})] 
        = \mathbb{E}_{\Lambda_{\tau}} [Q_{t}(x_{t}|\lambda_{\tau})],
         \forall \{\Lambda_{\tau},\Gamma_{\tau}\} \in \mathcal{L}_{\tau}.}\label{theorem1_equation}
    \end{align}         
\end{theorem}

\begin{proof}[Sketch of the proof]
    \todo{We first prove a one-step distribution-insensitive condition. To do so, we initially show that the derivative of the current action-value function equals the derivative of the future state-related cost, i.e., i.e., $q_{t}(x_{t}|\lambda_{t+1})=h_{t+1}(x_{t+1}(\lambda_{t+1}))$. Next, by taking the derivative of $h_{t+1}$ with regard to $\lambda_{t+1}$, we analyze the relationship between the derivative of the action-value function $q_{t}(x_t|\lambda_{t+1})$ and the future price $\lambda_{t+1}$. We then show that the derivative of the state-value function $v_t$ is a linear combination of all future functions $c_{\tau},\tau \in [t+1,T]$, indicating that, under the conditions specified in the theorem, $h_{t+1}(x_{t+1}(\lambda_{t+1}))$ is a linear function with regard to $\lambda_{t+1}$. Consequently, $q_t(x_t|\lambda_{t+1})$ is also linear in $\lambda_{t+1}$. We then prove that the linear relationship between $q_{t}(x_t|\lambda_{t+1})$ and $\lambda_{t+1}$ is both necessary and sufficient for satisfying the distribution-insensitive condition. Finally, we extend the one-step distribution-insensitive property to all future stages, thereby completing the proof of the theorem.}   
    The detailed proof is provided in the appendix.
\end{proof}

The theorem demonstrates that, under a quadratic state and action cost structure, the consumption decision of the demand model is independent of the future price distribution and depends only on its expectation. \todo{Specifically, (\ref{theorem1_equation}) shows that the action-value function $Q_t$ has the same expectation under different distributions, implying that the state-value function $V_t$, which captures the influence of future uncertainty at stage $t+1$, remains unchanged. Consequently, the demand model defined in (\ref{eq1}) at stage $t-1$ is equivalent under different future price distributions, leading to identical optimal solutions.} This result suggests that, in practice, even when future price uncertainty becomes more variable (while maintaining the same expectation), distribution-insensitive demand decisions in earlier stages remain unaffected.

\todo{In terms of risk-aware behavior, it is important to note that when state and action variables lie within a reasonable range, the price term $\lambda_t p_t$ dominates the objective function. As a result, from the perspective of cost, increasing price uncertainty or the presence of extreme price realizations may raise expected costs, which the demand model dislikes. However, by Theorem \ref{price_invariant_theorem}, the demand decisions do not change in response to this risk. This highlights a key limitation of the quadratic formulation: it is insufficient to capture true risk-aware decision behavior. To address this, we naturally extend the model to a superquadratic cost formulation to study richer behavioral responses under uncertainty, provided in the following corollary.}



\begin{corollary}\label{corollary_1}\emph{Distribution-sensitive demand models.}
    Consider the price distribution described in Theorem \ref{price_invariant_theorem}, given that either or both super quadratic state cost function $C_t$ and action cost function $G_t$ following Definitions \ref{def1} and \ref{def2}, the demand model in (\ref{eq1}) becomes distribution-sensitive, except an extra case that has symmetrical price distributions with a mean of zero and demand's prior state is zero, i.e., $\mathbb{E}_{\Lambda_{t+1}}[\lambda_{t+1}] =0, x_{t-1}=0$. \label{corollarystate} 
\end{corollary}
\begin{proof}
The proof is provided in the appendix.
\end{proof}
 
\todo{This corollary indicates that a superquadratic demand formulation leads to a distribution-sensitive demand model. It is worth noting that the exceptional case where distribution-insensitive behavior still holds under a superquadratic formulation, as mentioned in the corollary, is generally unrealistic in practice, since price distributions are typically skewed and have nonzero means. }

\todo{Moreover, we emphasize that superquadratic demand formulations are not only theoretically valid but also behaviorally realistic. Consumers often exhibit higher-order utility structures. For instance, thermal discomfort tends to increase disproportionately when temperatures become extremely hot or cold, and many energy-consuming devices, such as batteries, operate under hard physical constraints (e.g., capacity limits). 
These real-world considerations naturally challenge the adequacy of the traditional quadratic, distribution-insensitive model and underscore the importance of using superquadratic formulations to better understand and capture demand behavior under uncertainty, which we explore in detail in the next section.}


\section{Prudent demand models}
\todo{In this section, we study the distribution-sensitive demand model. Specifically, we introduce the concept of prudent demand, where the demand level adjusts proactively in response to future price distributions and expectations~\cite{skewness_left}. 
We show that when the state cost function $C_t$ is superquadratic and the action cost function $G_t$ is quadratic, the demand model exhibits prudence. We establish a connection between the future price distribution $\Lambda_{t+1}$ and the demand levels at earlier stages $p_{\tau},\tau \in [1,t]$. Furthermore, we reveal the skewness-aversion behavior of prudent demand: greater price skewness induces higher demand levels in earlier stages, and the magnitude of this adjustment increases with the degree of skewness.}

\todo{To analyze the prudent behavior of demand that is influenced by future price distributions, we first introduce special two-point price distributions as a foundational case. This serves as a basis for extending the framework to analyze more complex distributions.
Building upon this foundation, we then introduce the main Theorem.}
\begin{theorem}\emph{Prudent demand models}.\label{aggressive_theorem}
\begin{subequations}\label{price_model}
        Consider the price $\lambda_{t}$ come from two-point distributions $\Lambda_{t}(\gamma_t,\pi_t)$, $\Gamma_{t} (\gamma_t,\pi'_t)$
        at stage $t$ satisfying the following given $w_{\Lambda_t},w_{\Gamma_{t}} \in [0,0.5]$:
        \begin{align}     
            -\mathbb{E}[\lambda_t] < \gamma_t < \pi_t &<\pi'_t, \label{price_constraint}\\
            -(1-w_{\Lambda_t})\gamma_t + w_{\Lambda_t}\pi_t &= \mathbb{E}_{\Lambda_t}[\lambda_t], \label{price_old}\\
             -(1-w_{\Gamma_{t}})\gamma_t + w_{\Gamma_{t}}\pi'_t &= \mathbb{E}_{\Gamma_{t}}[\lambda_t],  \label{new_price} \\
            \mathbb{E}_{\Lambda_t}[\lambda_t] = \mathbb{E}_{\Gamma_{t}}[\lambda_t] &= \mathbb{E}[\lambda_t] \geq 0. 
        \end{align}
        Then a demand model described in (\ref{eq1}), with super quadratic state cost function $C_t$ and quadratic action cost function $G_t$ following Definitions \ref{def1} and \ref{def2}, under the conditions of $x_0 = 0$, $\lambda_{\tau} = 0, \tau \in [1,t]$, satisfy the \emph{prudence} and its sensitivity conditions that demand’s consumption before stage $t$ shows the following properties with regard to price distributions at stage $t+1$:
    \begin{align}
        \mathbb{E}_{\Gamma_{\tau+1}} [Q_{\tau}(x_{\tau}|\lambda_{\tau+1})] \leq \mathbb{E}_{\Lambda_{\tau+1}} [Q_{\tau}(x_{\tau}|\lambda_{\tau+1})] \leq Q_{\tau}(x_{\tau}|\mathbb{E}_{\Lambda_{\tau+1}} [\lambda_{\tau+1}]),\quad  \forall \tau \leq t. \label{prudent1} 
    \end{align}    
\end{subequations}
\end{theorem}

\begin{proof}[Sketch of the proof]
    This theorem shows that the future two-point price distribution with fixed expectations affects the current state-value function of the demand, and then affects the demands' actions in the earlier stages. Here, we study the right skewness condition, and the distribution $\Gamma$ is more right-skewed compared with the distribution $\Lambda$, and the left skewness follows the same analysis, which we show in the remark.
    The overview of this proof is to find the causal relationship between future price distribution and prior actions, with connecting variables and functions, i.e., 
    $\Lambda_{t+1} \sim q_{t} \sim v_t \sim x_t \sim p_{\tau,\forall \tau \in[1,t]}$. Specifically, we first connect future price distribution $\Lambda_{t+1}(\gamma_{t+1},\pi_{t+1})$ with the derivative of current state-value function $v_t$ by the model definition.
    Then, we connect the state and action $x_t, p_t$ to the future price $\gamma_{t+1},\pi_{t+1}$.
    This helps us rewrite the state-value function derivative $v_t$ as a function of future price. Then, we analyze the property of the reverse state-value function derivative $-v_t$, and include $\Gamma_{t+1}$ distributions to show the sensitivity of prudent demand.
    The proof consists of four general steps:
    \begin{itemize}
        \item \todo{We write out the extended form for the first stage optimization problem with future price and take the optimality conditions regarding actions~\cite{convex}. Then, we show that the derivative of the state-value function $v_t$ in the optimality conditions can be written as a function of future price.}
        \item From the optimality conditions of actions (states) at stage $t+1$, we reveal the function relationship between price $\gamma_{t+1}, \pi_{t+1}$ and state $x_{t+1}$, and show its monotonicity, symmetry, and concavity. 
        \item By using the price and state function property from the last step, we show that the reverse state-value function derivative $-v_t$ in the optimality conditions is positive and greater when using price distributions than when using price expectations, and the state value $x_t$ follows.
        \item To analyze the optimality conditions with a more right-skewed price distribution $\Gamma_{t+1}$, we further prove the sensitivity of the reverse state-value function derivative $-v_t$. We show a strict increase in $x_t$ with the right price skewness. Combined with the objective function and state transition, we show that all earlier stages' states and actions should be non-negative and non-decreasing with the skewness and complete the proof of Theorem \ref{aggressive_theorem}.
    \end{itemize}
    The full version of the proof is provided in the appendix.
\end{proof}

\begin{remark}\emph{Mirror theorem for left skewness.}\label{remark}
From Theorem \ref{aggressive_theorem}, demands show prudence with right-skewness price distributions $\Lambda, \Gamma$. As the system is linear and the cost function is symmetrical, the prudence and its sensitivity conditions also stand (reverse) for left skewed price distribution $\Gamma_{\mathrm{l}}(\gamma',\pi)$ and $\Lambda (\gamma,\pi)$ from Theorem \ref{aggressive_theorem}, i.e., $-\mathbb{E}[\lambda] <\gamma'<\gamma<\pi$, $-(1- w_{\mathrm{l}})\gamma' + w_{\mathrm{l}}\pi = \mathbb{E}[\lambda] \geq 0, w_{\Lambda},w_{\mathrm{l}}\in[0.5,1]$.
\end{remark}

This theorem highlights that our demand model, with time dependency, state transitions, and a high-order cost structure, a setting that naturally arises in human-centric decision-making contexts, exhibits prudent behavior. Specifically, when consumers are notified of an event with skewed price uncertainty at a future stage $t+1$, the action (demand) levels in all preceding periods either increase or remain unchanged, reflecting precautionary saving behavior. 
This also highlights that the state value change follows
$$
x_{\tau}^\Gamma \geq x_{\tau}^\Lambda \geq x_{\tau}^{\mathbb{E}_{\Lambda}} \geq 0, \quad \forall \tau \leq t,
$$
where for $\mathcal{D}\in\{\Gamma,\Lambda,\mathbb{E}_{\Lambda}\}$, $x_{\tau}^{\mathcal{D}}$ denotes the optimal demand at stage $\tau$ when the stage-$(t{+}1)$ price follows $\mathcal{D}$.
Importantly, the changes in demand levels are influenced by both the price expectation and the distributional characteristics. To eliminate the influence of normal (non-event) price fluctuations on the behavioral analysis, we normalize the price levels prior to the event to a reference value of zero. 
Our results align with the classical notion of prudence (Kimball~\cite{prudent_earliet}), whose defining condition $u'''>0$, where $u$ is the agent's utility over consumption, so that risk raises its expectation and induces precautionary behavior. Because our risk is carried by the price rather than additively by the state, the corresponding condition is convexity of the marginal state-value $-q_t(x_t\mid\lambda_{t+1})$ in the price $\lambda_{t+1}$---supplied by the super-quadratic state cost ($\partial^3 C_t/\partial x_t^3>0$ for $x_t>0$)---which reduces the expected marginal cost of holding $x_t>0$.


\todo{This preparatory behavior reflects an inherent risk-averse response in the demand model, rather than only bearing higher expected costs. Accordingly, our model offers a theoretical explanation for the first principle of risk-averse decision behavior and suggests that the commonly used risk-averse formulations are approximations of a more realistic, higher-order structure.
Moreover, greater skewness in the future price distribution leads to even higher demand levels in earlier periods, capturing the skewness aversion behavior. Specifically, when comparing more skewed price distributions $\Gamma$ with identical expectation, the state value $x_{t}$ strictly increases with the skewness of the distribution.}

\todo{In practice, the prudent behavior identified by our model suggests that pre-informing customers about upcoming uncertain events can enhance their preparatory actions, leading to increased savings in the state-dependent demand component and providing additional backup capacity to the system. Accordingly, operators should account for consumer behavior before the event when planning generation or curtailment strategies, rather than focusing solely on the event itself. For example, utilities may need to schedule additional generation in advance or implement price incentives to mitigate unintended demand peaks that arise from precautionary saving behavior ahead of the event. Moreover, regarding the skewness aversion behavior, operators should emphasize anticipating and preparing for tail-risk events, which have a greater impact on consumers' demand patterns. It also highlights the critical role of utilities in controlling the price ceiling to which consumers are exposed, as excessive exposure to highly variable prices can lead to unexpected demand patterns. For instance, utilities should avoid exposing retail consumers to variable prices, such as wholesale market prices, which often exhibit significant variation and skewness.}

Note that we only imply the decision behavior from stage 1 to $t$ before the event. We directly show the increase of the reverse state-value function derivative and the state value just ahead of the event (at stage $t$). We obtain actions and states in earlier stages from the cost function structures and the state transition process. Obviously, the state value increment at stage $t$ will be distributed to the prior stages. Specifically, with quadratic action cost and no discount rate ($A=1$), the state value increment is evenly distributed across all the prior stages for the demand model because the derivative of the quadratic cost is linear. If considering the discount rate ($A<1$), the increment is distributed more in the later stages, and the influence of price distribution on the state value will be eliminated after a limited state transitions. This motivates us to provide the following Corollary.
\begin{corollary}\emph{Strict conditions}.
    Considering the price distributions and demand model described in Theorem \ref{aggressive_theorem} with the discount rate $A<1$, there exists a stage $\tau_0$, $\tau_0 < t$, such that the prudent conditions become strictly stands: 
    \begin{align}
        \mathbb{E}_{\Gamma_{\tau+1}} [Q_{\tau}(x_{\tau}|\lambda_{\tau+1})]  < \mathbb{E}_{\Lambda_{\tau+1}} [Q_{\tau}(x_{\tau}|\lambda_{\tau+1})] < Q_{\tau}(x_{\tau}|\mathbb{E}_{\Lambda_{\tau+1}} [\lambda_{\tau+1}]), \quad 
        \forall \tau_0 < \tau \leq t.\label{strictly_conditions}
    \end{align}  
\end{corollary}

\begin{proof}
   We provide an intuitive proof of this corollary. Due to the discount factor $A$ with $A<1$, the influence of the state value at stage $t$ diminishes gradually as it propagates backward through earlier stages. This degradation eventually results in a zero action after a certain number of transitions (denote this stage as $t-\tau_0$), where $x_{\tau} \approx 0$ for all $\tau \leq \tau_0$. The stage $\tau_0$ approximately satisfies:
    \begin{align}
        x_{\tau_0} \approx A^{t-\tau_0} x_t.
    \end{align}
    
   This relationship reflects the degradation of the state-value function through transitions. When $A^{t-\tau_0} \approx 0$, we have $x_{\tau_0} \approx 0$, implying that the influence of the event occurring at stage $t+1$ is effectively eliminated after $t-\tau_0$ transitions. As a result, the action $p_{\tau},\forall \tau \leq \tau_0$ will not be affected by the price uncertainty at stage $t+1$.
    Thus, during stages $(\tau_0,t]$, we conclude that the strict prudent conditions for the demands hold.
\end{proof}

\todo{This corollary establishes the strict conditions under which prudent demand behavior arises, specifically identifying the exact periods where demand levels change, typically shortly before the event occurs. Beyond the extrapolation of strict conditions, we also note that Theorem \ref{aggressive_theorem} characterizes \emph{prudent} behavior and \emph{its sensitivity conditions} under a two-point price distribution. We subsequently extend the analysis to address more complex price distributions with fixed expectations. The prudent condition can be naturally generalized to arbitrary skewed price distributions. This is because any skewed distribution can be discretized into a combination of two-point price distributions with the same expectation, and each of these two-point pairs satisfies the prudent condition outlined in Theorem \ref{aggressive_theorem}. Thus, the overall demand model still exhibits prudence by aggregating these two-point distributions. We then provide an example in the following corollary to extend the sensitivity conditions of prudent demand to more complex price distributions.}

\begin{corollary}\emph{Prudent demand sensitivity extrapolation.} \label{distribution} 
\begin{subequations}\label{complicated_pricemodel}
Consider two price distributions $\Lambda_t$ and $\Gamma_t$ at stage $t$ with probability density functions (PDFs) $f_{\Lambda_t}(\lambda_t)$ and $f_{\Gamma_t}(\lambda_t)$. Denote variable $\lambda_t$ as $\gamma_t,\pi_{\Lambda,t}$ for distribution $\Lambda_t$, $\gamma_t<\mu_t,\pi_{\Lambda,t}\geq \mu_t$, and $\gamma_t,\pi_{\Gamma,t}$ for distribution $\Gamma_t$, $\gamma_t<\mu_t,\pi_{\Gamma,t}\geq \mu_t$. Then
defined over a real set $\mathcal{X}_t$ satisfying
\begin{align}
 \mathbb{E}_{\Gamma_t}[\gamma_t] + \mathbb{E}_{\Gamma_t}[\pi_{\Gamma,t}] 
 =  \mathbb{E}_{\Lambda_t}[\gamma_t] + \mathbb{E}_{\Lambda_t}[\pi_{\Lambda,t}] = \mu_t \geq 0 \label{complicated_pricemodel1},\\
    f_{\Gamma_t}(\gamma_t) > f_{\Lambda_t}(\gamma_t), \forall \gamma_t \in \mathcal{X}_t, \label{complicated_pricemodel2}\\
    f_{\Gamma_t}(\pi_{\Gamma,t}) \leq f_{\Lambda_t}(\pi_{\Lambda,t}),  \forall \pi_{\Gamma,t} > \pi_{\Lambda,t}, \{\pi_{\Gamma,t}, \pi_{\Lambda,t}\} \in \mathcal{X}_t, \\
    f_{\Gamma}(\pi_{\Gamma})+f_{\Gamma}(\gamma) < f_{\Lambda}(\pi_{\Lambda})+f_{\Lambda}(\gamma), \forall \{\pi_{\Gamma,t}, \pi_{\Lambda,t}\} \in \mathcal{X}_t.\label{seperate_condition}\\
    f_{\Gamma_t}(\pi_{\Gamma,t}) \leq f_{\Gamma_t}(\gamma_t),  \forall \{\pi_{\Gamma,t}, \gamma_t\} \in \mathcal{X}_t,  \\
    f_{\Lambda_t}(\pi_{\Lambda,t}) \leq f_{\Lambda_t}(\gamma_t),  \forall \{\pi_{\Lambda,t}, \gamma_t\} \in \mathcal{X}_t, \label{complicated_pricemodel5} 
\end{align}
Then, it is sufficient for the demand model to show the sensitivity of prudence as described in Theorem \ref{aggressive_theorem},
\begin{align}
    \mathbb{E}_{\Gamma_{t+1}}[Q_t(x_t|\lambda_{t+1})] < \mathbb{E}_{\Lambda_{t+1}}[Q_t(x_t|\lambda_{t+1})] 
\end{align}
\end{subequations}
\end{corollary}

\begin{proof}
    The proof is provided in the appendix.
\end{proof}

This Corollary extrapolates the prudent demand sensitivity condition from a two-point price distribution to a continuous price distribution satisfying given conditions. Also, with discrete price distributions, by applying the PDF conditions (\ref{complicated_pricemodel}) to the probability mass function (PMF), the prudent demand sensitivity condition still holds.

\section{Case Study}
\todo{In this section, we first introduce an illustrative example and conduct a numerical simulation to demonstrate the prudent behavior of the demand model and perform a sensitivity analysis to investigate how key parameters influence the prudent outcomes. Following this, we apply our framework to a real-world case study to verify the effectiveness of the prudent demand formulation.}

\todo{We set the terminal state-value function to zero and assume a discount factor of $A=1$ unless otherwise specified. The parameter of the quadratic action cost function is set to $a_{\mathrm{p}}=1$. We adopt a log-barrier function to represent the state cost function, which imposes steep penalties near the boundaries. The boundary limit is set as $x_{\mathrm{max}} = 20$. The state cost function is defined as follows:}
\begin{subequations}
\begin{align}
    C_t(x_t) &= -\alpha_{\mathrm{c}} \ln(x_{\mathrm{max}} - x_{t}) - \alpha_{\mathrm{c}} \ln(x_{\mathrm{max}} + x_{t}) + 2\alpha_{\mathrm{c}} \ln x_{\mathrm{max}}, \label{state_cost}
    \end{align}
    \todo{where $\alpha_{\mathrm{c}}$ is a function parameter, set to 0.5. The first and second term indicates the upper and lower bounds of $x_t$, respectively, and the third term is the regularization to ensure $C(0)=0$. Taking the derivative of the state cost function with respect to $x_t$, we have
    \begin{align}
    c_t(x_t) &= \frac{\partial C_t(x_t)}{\partial x_t} = \frac{\alpha_{\mathrm{c}}}{x_{\mathrm{max}} - x_{t}} - \frac{\alpha_{\mathrm{c}}}{x_{\mathrm{max}} + x_{t}}.
\end{align}}
\end{subequations}

\subsection{An illustration example}
We first use a two-stage illustration example with a two-point price distribution $\gamma,\pi$, with fixed expectation $-(1-w)\gamma + w\pi = \mathbb{E}[\lambda] \geq 0$, and $w\in[0,0.5], \pi>\gamma>0$, to show the performance of the prudent demand. We set the event to happen at the 2nd stage and analyze the 1st stage's state value and demand level. The objective function is: 
\begin{subequations}
\begin{align}
    Q_0(x_0|\lambda_1,\gamma,\pi) &=\min_{p_1,p_{\mathrm{\pi},2},p_{\mathrm{\gamma},2}} 
    [\lambda_{1} p_{1} + \frac{a_{\mathrm{p}}}{2}p_{1}^2 - \alpha_{\mathrm{c}} \ln(x_{\mathrm{max}} - x_{1})
    - \alpha_{\mathrm{c}} \ln(x_{\mathrm{max}} + x_{1}) 
    + 2 \alpha_{\mathrm{c}} \ln x_{\mathrm{max}} ]   \nonumber \\
    &+w[\pi p_{\mathrm{\pi},2} + \frac{a_{\mathrm{p}}} {2}p_{\mathrm{\pi},2}^2 -
    \alpha_{\mathrm{c}} \ln(x_{\mathrm{max}} - x_{\mathrm{\pi},2}) 
    -\alpha_{\mathrm{c}} \ln(x_{\mathrm{max}} + x_{\mathrm{\pi},2}) + 2 \alpha_{\mathrm{c}} \ln x_{\mathrm{max}}]  \nonumber \\
    &+(1-w) [-\gamma p_{\mathrm{\gamma},2} + \frac{a_{\mathrm{p}}} {2}p_{\mathrm{\gamma},2}^2 
    -\alpha_{\mathrm{c}} \ln(x_{\mathrm{max}} 
    - x_{\mathrm{\gamma},2}) 
    - \alpha_{\mathrm{c}} \ln(x_{\mathrm{max}} + x_{\mathrm{\gamma},2}) + 2 \alpha_{\mathrm{c}} \ln x_{\mathrm{max}}] \\
    \text{s.t. } &x_1 = x_0 + p_1,\ 
    x_{\mathrm{\gamma},2} = x_1 + p_{\mathrm{\gamma},2}, \ 
    x_{\mathrm{\pi},2} = x_1 + p_{\mathrm{\pi},2}. \label{illustra1}
\end{align}
\end{subequations}
Taking the optimality conditions with regards to $p_1,p_{\mathrm{\gamma},2},p_{\mathrm{\pi},2}$ and take $\lambda_1= 0$ and $x_0=0$ inside, 
\begin{subequations}
    \begin{align}
        p_1: &a_{\mathrm{p}} p_1 + \frac{\alpha_{\mathrm{c}}}{x_{\mathrm{max}} - x_{1}} - \frac{\alpha_{\mathrm{c}}}{x_{\mathrm{max}}+ x_{1}} 
        +w [\frac{\alpha_{\mathrm{c}}}{x_{\mathrm{max}} - x_{\mathrm{\pi},2}} - \frac{\alpha_{\mathrm{c}}}{x_{\mathrm{max}}+ x_{\mathrm{\pi},2}} ] \nonumber\\
        &+(1-w) [\frac{\alpha_{\mathrm{c}}}{x_{\mathrm{max}} - x_{\mathrm{\gamma},2}} - \frac{\alpha_{\mathrm{c}}}{x_{\mathrm{max}}+ x_{\mathrm{\gamma},2}}] =0  \\
        p_{\mathrm{\pi},2}: &\pi + a_{\mathrm{p}} p_{\mathrm{\pi},2} +  \frac{\alpha_{\mathrm{c}}}{x_{\mathrm{max}} - x_{\mathrm{\pi},2}} - \frac{\alpha_{\mathrm{c}}}{x_{\mathrm{max}}+ x_{\mathrm{\pi},2}} =0 \\
        p_{\mathrm{\gamma},2}: &-\gamma + a_{\mathrm{p}} p_{\mathrm{\gamma},2} + \frac{\alpha_{\mathrm{c}}}{x_{\mathrm{max}} - x_{\mathrm{\gamma},2}} - \frac{\alpha_{\mathrm{c}}}{x_{\mathrm{max}}+ x_{\mathrm{\gamma},2}} =0 
    \end{align}
    among them, (\ref{illustra1}) also stands.
\end{subequations}

We set $\gamma = 1$, and gradually increase $\pi$ starting from 1, while maintaining $\mathbb{E}[\lambda] = 0$ to observe the influence of the price distribution. Fig. \ref{illu} shows the state value at the 1st stage. \todo{When $\pi=1$, the price distribution is symmetric, resulting in a state value of zero. As $\pi$ increases, the price distribution becomes skewed, and the state value shifts positively. Moreover, the state value increases as $\pi$ grows, demonstrating prudent demand's skewness-aversion behavior.}
\begin{figure}
\centering
\begin{minipage}{0.45\textwidth}
  \centering
  \includegraphics[width=0.75\textwidth]{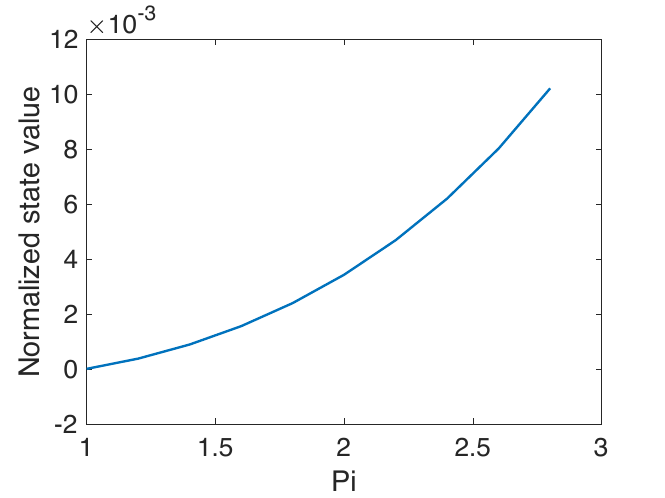}
  \caption{State value at 1st stage.}
  \label{illu}
\end{minipage}\hfill
\begin{minipage}{0.45\textwidth}
  \centering
  \includegraphics[width=0.75\textwidth]{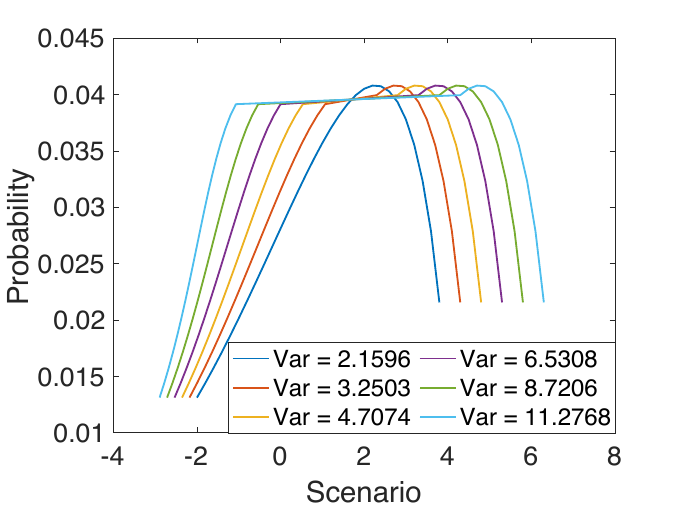}
  \caption{Uncertain price distribution.}
  \label{Con_distri}
\end{minipage}
\end{figure}

\subsection{Prudent demand}
We design a case study with 24 stages, i.e., $T=24$ (time slots), to verify the theoretical analysis of prudent demand behavior. In this setup, we construct six skewed price distributions, each with increasing variance but the same expectation. We assume that the uncertain event occurs at the 10th stage. With a fixed expectation and skewed distribution, greater variance in the price distribution corresponds to greater skewness. We discretize each price distribution into 30 scenarios with the occurrence probability of each scenario. Fig. \ref{Con_distri} shows the detailed price distributions.

\todo{We solve the problem using the stochastic dynamic programming framework described in (\ref{eq1}). Figure \ref{continuous}(a) shows the error curve with respect to the iteration number. The error is measured by the $l_2$-norm of the state value differences between successive iterations. The algorithm is considered converged once the state values stabilize. As shown, the error curve exhibits good convergence behavior, and the computation time is approximately 2 seconds for solving one price distribution.
Figure \ref{continuous}(b) presents the 24-stage state values under the first price distribution. The state values from stages 1 to 9 gradually increase until the event at stage 10, indicating changes in the demand level across all prior stages. This pattern reflects the risk-averse behavior of prudent demand, where the system proactively prepares before the uncertain event occurs. This insight suggests that operators should schedule generation or curtailment plans in advance to accommodate the proactive behavior of consumers. Alternatively, operators could consider announcing potential uncertain events ahead of time to leverage consumers' preparatory actions, rather than focusing only on responses when the event occurs. }


\todo{Figure \ref{continuous}(c) illustrates the sensitivity condition of prudent demand. First, the state value at all prior stages increases across all price distributions. More importantly, the prudent demand shows skewness-averse behavior as the rate of increase (i.e., the slope) of the state value becomes steeper as the skewness of the price distribution increases. Theoretically, this behavior arises from the superquadratic cost function structure, where increases in the state value lead to disproportionately higher costs. In practical terms, this suggests that tail-risk events have a greater impact on changes in consumers' demand patterns. Accordingly, operators should focus more on tail-risk events and carefully manage the variance and price ceilings to which consumers are exposed to mitigate unexpected changes in consumers' demand patterns.}
\begin{figure}
\centering
\subfigure[\todo{Convergence curve under 1st price distribution}]{\includegraphics[width=0.32\textwidth]{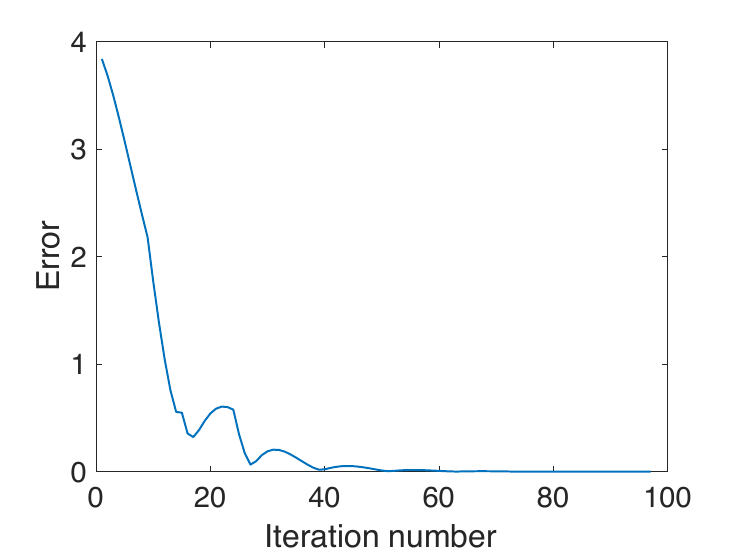}}
\subfigure[\todo{Normalized state value under 1st price distribution}]{\includegraphics[width=0.32\textwidth]{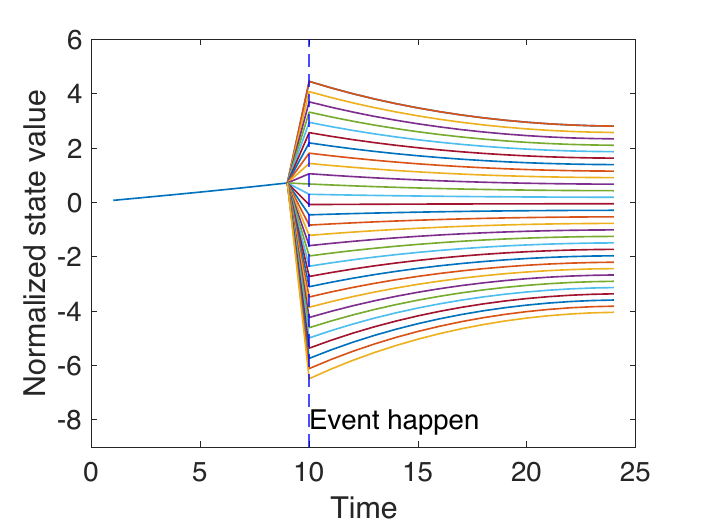}}
\subfigure[Normalized state value under all price distributions]{\includegraphics[width=0.32\textwidth]{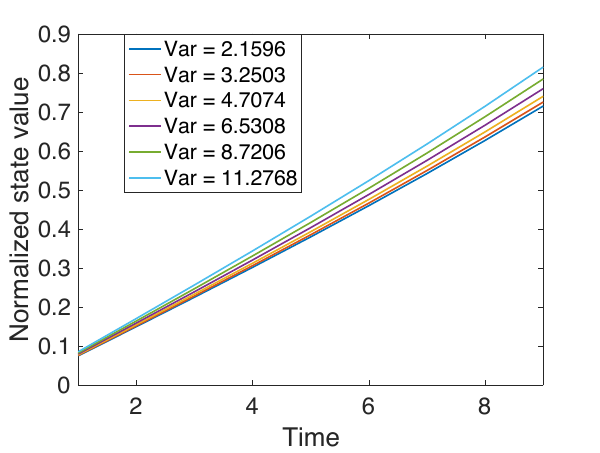}}
\caption{\todo{Prudent demand simulation results.}}
\label{continuous}
\end{figure}

\subsection{\todo{Sensitivity analysis of model parameters}}
\todo{We conduct a sensitivity analysis regarding key model parameters, specifically the discount factor $A$ and the state cost function parameter $\alpha_c$. These two parameters are selected because analyzing the action cost function parameter would yield similar insights to analyzing the state cost function parameter, and the superquadratic nature of the state cost function imposes a more significant effect on behavior. Moreover, varying the boundary limit of the state variable in the log-barrier function does not provide additional insight, as it is intuitively clear that a higher boundary would simply allow greater variation in the state value. Otherwise, since we normalize both power and state costs, and set the non-event price to zero to isolate the event’s influence, there are no other parameters that merit detailed analysis. For the following analysis, we fix the second price distribution as the reference setting.}

We vary the discount factor $A$ from 0.1 to 1 in increments of 0.1, and present the corresponding state values for stages 1 to 9 in Fig.~\ref{sensitivity}(a). As expected, a higher discount factor results in less loss during state transitions, leading to a greater accumulation of state value across stages. This, in turn, causes the influence of the event at stage 10 to have a stronger impact on earlier stages, as the effects of future uncertainty propagate more easily backward through the decision stages. Conversely, under a very low discount factor (e.g., $A = 0.1$), the influence of future uncertainty cannot effectively propagate to previous stages, thereby diminishing the precautionary saving behavior. These results indicate that prudent demand decision behavior is highly sensitive to the discount factor. Therefore, operators should carefully consider the discount structure when scheduling generation or designing curtailment plans in anticipation of prudent consumer responses.

\todo{We then vary the state cost function parameter $\alpha_c$ from 0.1 to 2.2 in increments of 0.3, and present the corresponding state values in Fig.~\ref{sensitivity}(b). Unlike the case with the discount factor, precautionary saving behavior persists across all parameter settings. Higher values of $\alpha_c$ imply greater potential loss from uncertainty at the event time for the same state value, thereby inducing stronger precautionary saving behavior. However, a saturation effect emerges: beyond a certain threshold, further increases in $\alpha_c$ do not significantly amplify savings. This is because, at very high parameter values, the cost of precautionary saving begins to outweigh the cost of expected risk at the event time. Therefore, operators should also carefully consider the magnitude of the state cost parameter when evaluating or quantifying the degree of precautionary savings.}

\begin{figure}
\centering
\subfigure[\todo{Normalized state value under different discount factor $A$}]{\includegraphics[width=0.4\textwidth]{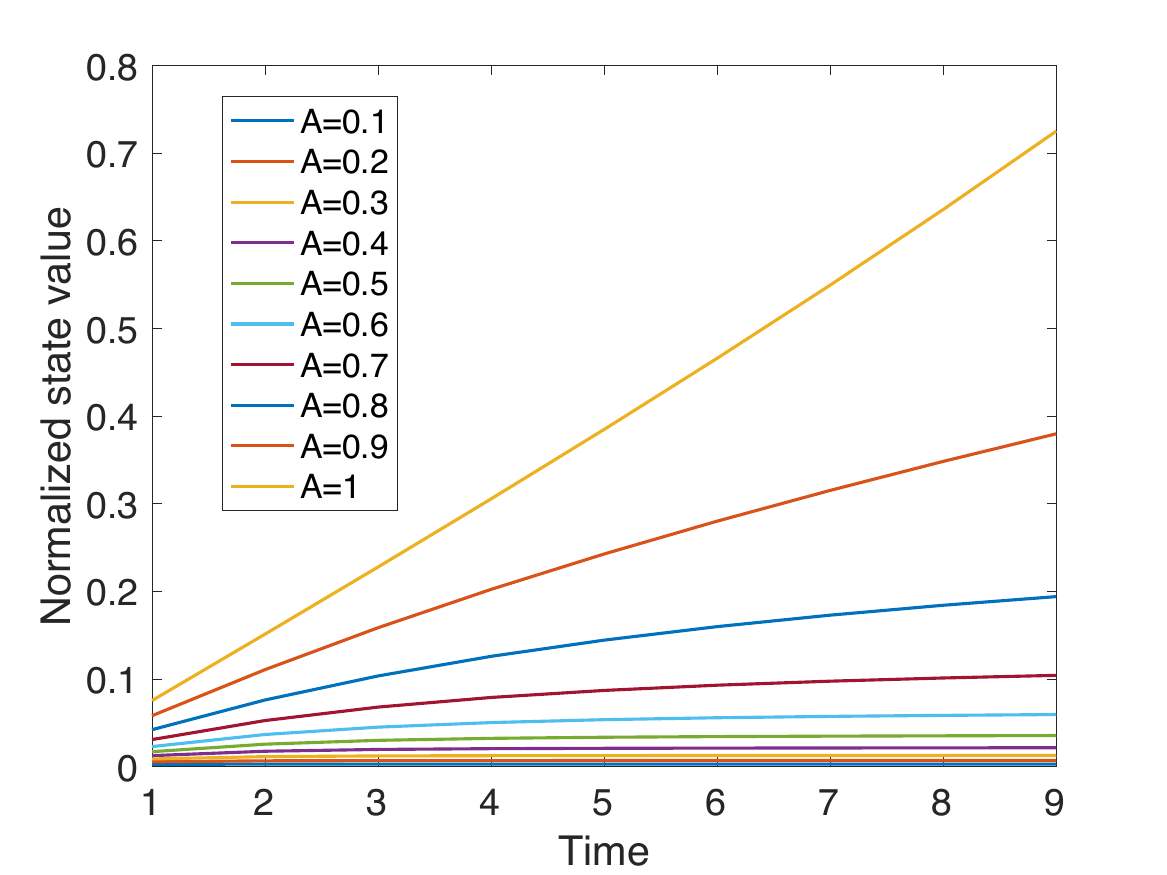}}
\subfigure[\todo{Normalized state value under different state cost function parameter $\alpha_c$}]{\includegraphics[width=0.4\textwidth]{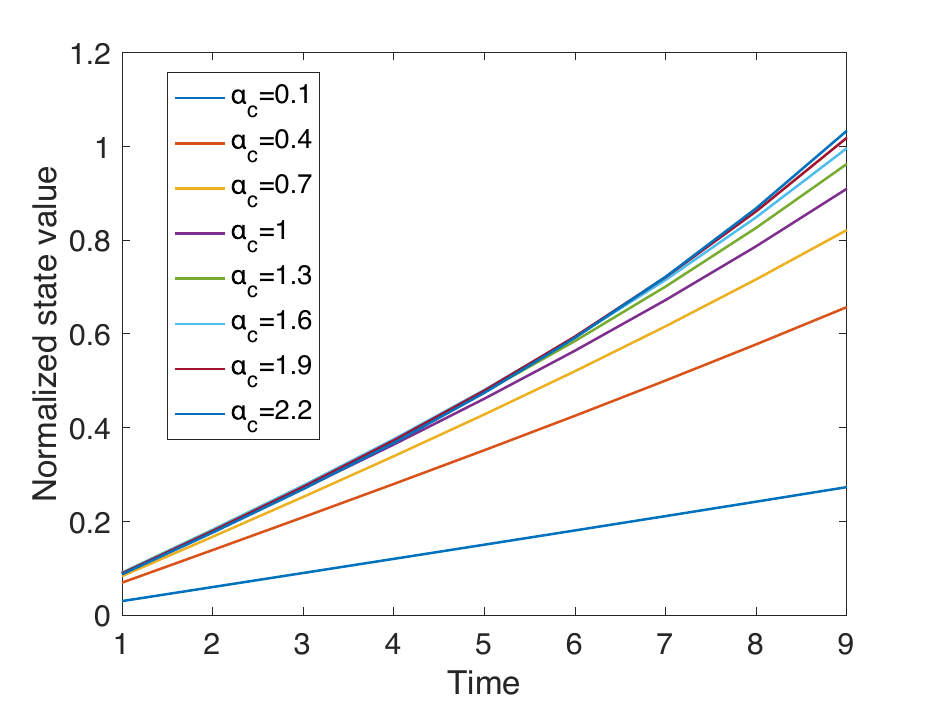}}
\caption{\todo{Sensitivity analysis of model parameters.}}
\label{sensitivity}
\end{figure}

\subsection{\todo{Real-world applications}}
\todo{In this section, we explore a real-world application scenario to demonstrate prudent demand behavior, specifically by comparing the performance of two modeling approaches: the quadratic demand model and the superquadratic demand model. Our goal is to show that the superquadratic formulation naturally induces prudent decision behavior. We adopt the modeling framework and dataset from a previous study~\cite{xu2020operational}, which uses real-world data from the New York independent system operator (NYISO). The dataset includes day-ahead (DA) prices and real-time (RT) prices from the year 2018, as well as a probability forecast of RT prices on February 1, based on historical DA–RT price biases observed during January 2018. Note that although we have access to the realized RT prices, we assume that RT prices are correlated with DA prices and can be partially inferred from them. Therefore, we use the empirical distribution of DA–RT differences from January to construct the distribution of price uncertainty for forecasting RT prices on February 1. The empirical distribution of DA-RT price differences is shown in Fig.~\ref{real_case}(a), which confirms the skewed nature of real-world price uncertainty.}

\todo{The model focuses on a battery system, treated as a specific form of demand, that performs arbitrage to maximize profit under non-anticipatory price uncertainty. We refer to the superquadratic model as the one used in the previous study~\cite{xu2020operational}, where hard bounds on charging/discharging rates and the state of charge (SoC) naturally result in a superquadratic cost structure. Hence, the analytical approach used previously to compute the value function derivative is directly applicable. To compare, we develop a relaxed quadratic model that introduces quadratic penalty terms in the objective for both charging/discharging actions and SoC levels, following the structure defined in Definition~\ref{def2}, with penalty coefficients set to 5. Using the same dual decomposition method, we derive a fixed-point equation in the dual variable that reflects the marginal value of SoC with penalty, and solve it numerically for each price realization. By taking the expectation over the price distribution, we obtain the value function derivative under uncertainty.}

\todo{Fig.~\ref{real_case}(b) presents the marginal value of SoC for both the superquadratic and quadratic models. The results show that large variance occurs during the evening times 20-21. Under the superquadratic model, the marginal value of SoC exhibits greater variability overall, particularly showing an increasing trend (mostly positively skewed) compared to the quadratic model, reflecting the prudent behavior. The quadratic model is more likely to capture the DA price, which corresponds to the mean value of RT price uncertainty. Specifically, the maximum marginal value of the SoC increases significantly prior to the high variance observed at time 20. Although the exact pattern may also be affected by intertemporal uncertainty and mutual price-state dependencies, this real-world application roughly shows superquadratic formulation provides richer and more responsive decision behavior. These findings suggest that incorporating superquadratic modeling allows operators to better identify and anticipate consumer demand patterns in the presence of uncertain future prices.}

\begin{figure}
\centering
\subfigure[\todo{Price distribution from NYISO}]{\includegraphics[width=0.4\textwidth]{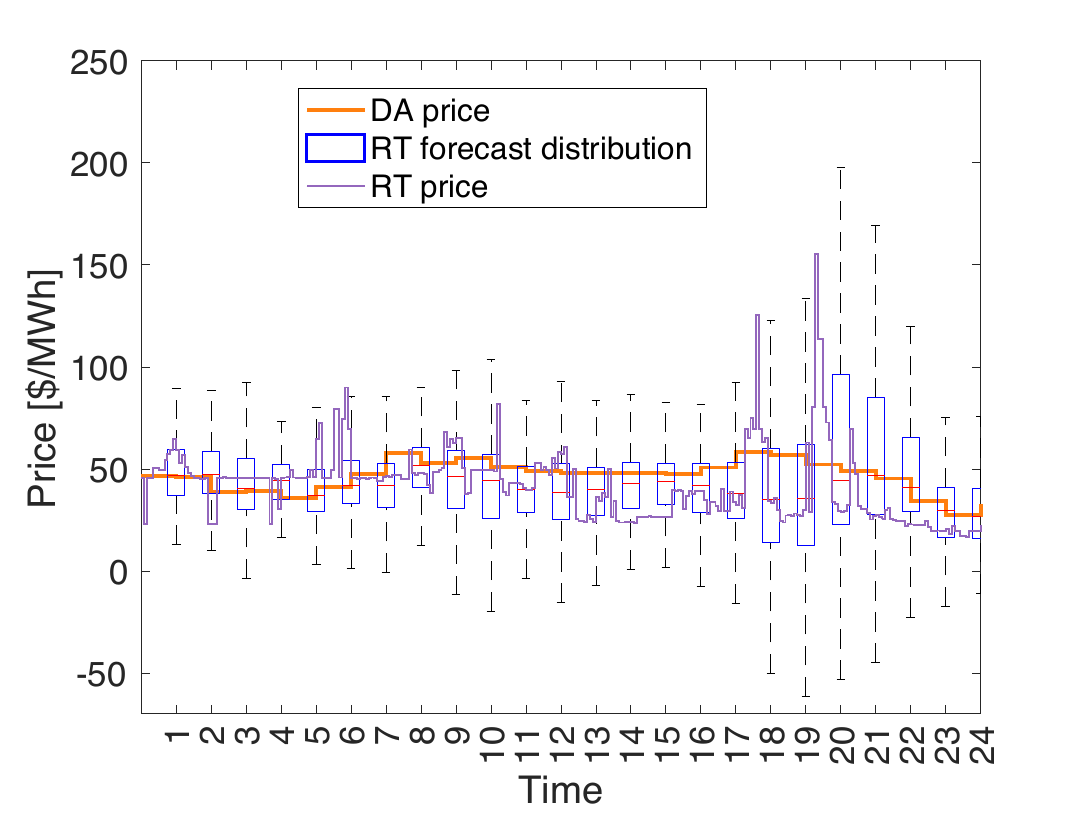}}
\subfigure[\todo{Battery SoC marginal value range under quadratic and superquadratic formulation}]
{\includegraphics[width=0.4\textwidth]{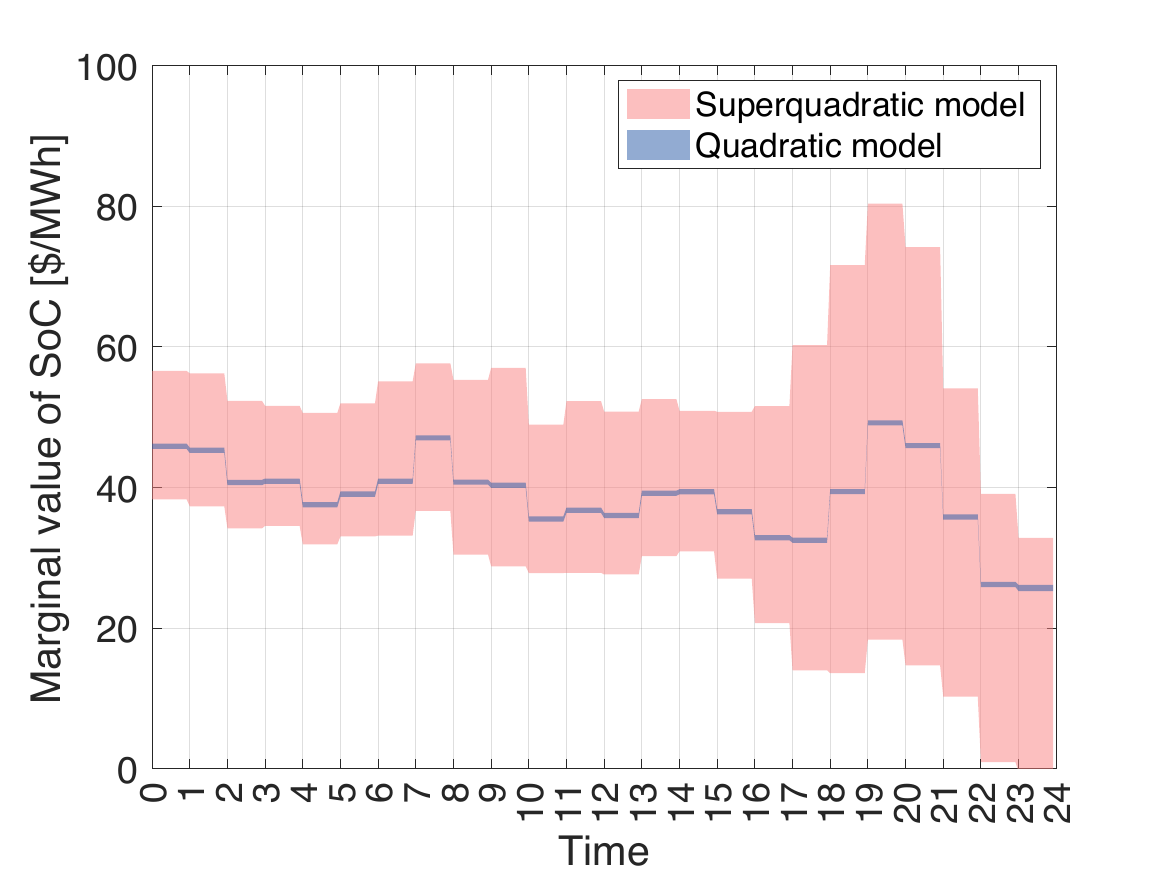}}
\caption{\todo{Real-world application.}}
\label{real_case}
\end{figure}

\section{Discussion and Conclusion}
\todo{We propose a theoretical framework to model demand decision behavior that combines non-anticipatory multi-stage stochastic decision making with non-quadratic cost functions. The framework reveals that demand with quadratic state cost functions, such as thermal discomfort, exhibits distribution-insensitive behavior, its action unaffected by changes in the distribution shape under fixed expectations. In contrast, demand behavior becomes prudent when the state cost formulation is superquadratic, a condition naturally arising in human-centric decision-making. Furthermore, we prove prudent demand exhibits skewness aversion, where more skewed distributions motivate greater precautionary savings. These results uncover fundamental behavioral properties in demand response and offer several insights for practical applications.} 

\todo{From a practical perspective, our findings suggest that practitioners and policy makers should adopt more sophisticated demand models, either through higher-order utility function formulations or the adoption of more accurate risk terms, especially considering physical and behavioral response limitations. Doing so can more accurately capture real demand behavior and better account for demand pattern changes to efficiently design time-varying tariffs. We also highlight that prudent behavior naturally leads to preparatory savings in the state-dependent components of demand, thereby providing additional backup capacity to the system ahead of emergencies. Consequently, operators should not only focus on the event time itself but also schedule additional generation or implement price incentives in advance to prevent unintended demand peaks. We also verify through numerical simulation that discounting factors and the sensitivity of demand to state value changes are important for accurately quantifying prudent behavior. Finally, our results emphasize the importance of anticipating and preparing for tail-risk events and controlling the price ceilings exposed to customers, for example, avoiding the direct exposure of retail consumers to highly volatile wholesale market prices, to maintain stable and controllable system operations. }

\todo{Future research will extend this framework by moving from the individual demand model to the agent setting, investigating whether prudent behavior persists when consumers act as price-makers rather than price-takers, and formulating decision frameworks incorporating strategic interactions among agents within a multi-agent decision-making environment.}


\bibliographystyle{apalike}
\bibliography{ref}

@article{pennings2003shape,
  title={The shape of utility functions and organizational behavior},
  author={Pennings, Joost ME and Smidts, Ale},
  journal={Management Science},
  volume={49},
  number={9},
  pages={1251--1263},
  year={2003},
  publisher={INFORMS}
}

@article{schuhmacher2021justifying,
  title={Justifying mean-variance portfolio selection when asset returns are skewed},
  author={Schuhmacher, Frank and Kohrs, Hendrik and Auer, Benjamin R},
  journal={Management Science},
  volume={67},
  number={12},
  pages={7812--7824},
  year={2021},
  publisher={INFORMS}
}

@article{ansarin2022review,
  title={A review of equity in electricity tariffs in the renewable energy era},
  author={Ansarin, Mohammad and Ghiassi-Farrokhfal, Yashar and Ketter, Wolfgang and Collins, John},
  journal={Renewable and Sustainable Energy Reviews},
  volume={161},
  pages={112333},
  year={2022},
  publisher={Elsevier}
}

@article{horowitz2014equity,
  title={Equity in residential electricity pricing},
  author={Horowitz, Shira and Lave, Lester},
  journal={The Energy Journal},
  volume={35},
  number={2},
  pages={1--24},
  year={2014},
  publisher={SAGE Publications Sage CA: Los Angeles, CA}
}

@article{aid2022optimal,
  title={Optimal electricity demand response contracting with responsiveness incentives},
  author={A{\"\i}d, Ren{\'e} and Possama{\"\i}, Dylan and Touzi, Nizar},
  journal={Mathematics of Operations Research},
  volume={47},
  number={3},
  pages={2112--2137},
  year={2022},
  publisher={INFORMS}
}

@article{lejeune2024profit,
  title={Profit-based unit commitment models with price-responsive decision-dependent uncertainty},
  author={Lejeune, Miguel A and Dehghanian, Payman and Ma, Wenbo},
  journal={European Journal of Operational Research},
  volume={314},
  number={3},
  pages={1052--1064},
  year={2024},
  publisher={Elsevier}
}

@article{powell2019unified,
  title={A unified framework for stochastic optimization},
  author={Powell, Warren B},
  journal={European journal of operational research},
  volume={275},
  number={3},
  pages={795--821},
  year={2019},
  publisher={Elsevier}
}

@article{oum2006hedging,
  title={Hedging quantity risks with standard power options in a competitive wholesale electricity market},
  author={Oum, Yumi and Oren, Shmuel and Deng, Shijie},
  journal={Naval Research Logistics (NRL)},
  volume={53},
  number={7},
  pages={697--712},
  year={2006},
  publisher={Wiley Online Library}
}

@article{zhang2020supply,
  title={Supply chains involving a mean-variance-skewness-kurtosis newsvendor: Analysis and coordination},
  author={Zhang, Juzhi and Sethi, Suresh P and Choi, Tsan-Ming and Cheng, TCE},
  journal={Production and Operations Management},
  volume={29},
  number={6},
  pages={1397--1430},
  year={2020},
  publisher={SAGE Publications Sage CA: Los Angeles, CA}
}

@article{ren2024risk,
  title={Risk-aversion versus risk-loving preferences in nonparametric frontier-based fund ratings: A buy-and-hold backtesting strategy},
  author={Ren, Tiantian and Kerstens, Kristiaan and Kumar, Saurav},
  journal={European Journal of Operational Research},
  volume={319},
  number={1},
  pages={332--344},
  year={2024},
  publisher={Elsevier}
}

@article{ebert2024first,
  title={First-Order Prudence and Its Implications for Precautionary Savings and the Risk-Free Rate},
  author={Ebert, Sebastian and Karehnke, Paul},
  journal={Operations Research},
  year={2024},
  publisher={INFORMS}
}

@article{nalpas2017portfolio,
  title={Portfolio selection in a multi-moment setting: A simple Monte-Carlo-FDH algorithm},
  author={Nalpas, Nicolas and Simar, L{\'e}opold and Vanhems, Anne},
  journal={European Journal of Operational Research},
  volume={263},
  number={1},
  pages={308--320},
  year={2017},
  publisher={Elsevier}
}

@misc{PJM,
   title = {DR Opportunities for End-Use customers in Wholesale Market},
 author = {{PJM Demand Response}},
   howpublished = {\url{https://www.pjm.com/-/media/DotCom/markets-ops/dsr/end-use-customer-fact-sheet.pdf}},
year = {2021}
}

@misc{4CP_ercot,
   title = {Overview of Demand
Response in ERCOT},
 author = {Kenan Ögelman},
   howpublished = {\url{https://www.ercot.com/files/docs/2023/05/19/ERCOT_Demand_Response__Summary_Spring_2023-update.pdf}},
year = {2016}
}

@misc{PJM_CP,
   title = {PJM’s Economic Demand response Program},
 author = {{CPower}},
   howpublished = {\url{https://cpowerenergy.com/wp-content/uploads/2017/08/PJM_Economic_Demand_Response_Program.pdf}},
year = {2017}
}

@inproceedings{baosen_cdc,
  title={Mitigation of Coincident Peak Charges via Approximate Dynamic Programming},
  author={Dowling, Chase P and Zhang, Baosen},
  booktitle={2019 IEEE 58th Conference on Decision and Control (CDC)},
  pages={4202--4207},
  year={2019},
  organization={IEEE}
}

@misc{4CP,
   title = {4CP Management System - A DYNAMIC SOLUTION TO MANAGING 4CP CHARGES},
 author = {{CPower}},
   howpublished = {\url{https://cpowerenergy.com/wp-content/uploads/2016/12/ERCOT_4CP_Web_Download.pdf}},
year = {2016}
}

@misc{NREL,
   title = {NREL Shows Live Grid Impacts From Total Solar Eclipse},
 author = {Moriah Petty},
   howpublished = {\url{https://www.nrel.gov/news/program/2024/nrel-shows-live-grid-impacts-from-the-total-solar-eclipse.html}},
year = {2024}
}

@misc{tesla,
   title = {Storm Watch},
note = {Retrieved Apr 27},
 author = {{Tesla}},
   howpublished = {\url{https://www.tesla.com/support/energy/powerwall/mobile-app/storm-watch?utm_source=chatgpt.com}},
year = {2025}
}

@article{DR_intro_earlist,
  title={Autonomous demand-side management based on game-theoretic energy consumption scheduling for the future smart grid},
  author={Mohsenian-Rad, Amir-Hamed and Wong, Vincent WS and Jatskevich, Juri and Schober, Robert and Leon-Garcia, Alberto},
  journal={IEEE transactions on Smart Grid},
  volume={1},
  number={3},
  pages={320--331},
  year={2010},
  publisher={IEEE}
}

@book{shapiro2021lectures,
  title={Lectures on stochastic programming: modeling and theory},
  author={Shapiro, Alexander and Dentcheva, Darinka and Ruszczynski, Andrzej},
  year={2021},
  publisher={SIAM}
}

@inproceedings{optimal_response_GM,
  title={Optimal demand response based on utility maximization in power networks},
  author={Li, Na and Chen, Lijun and Low, Steven H},
  booktitle={2011 IEEE power and energy society general meeting},
  pages={1--8},
  year={2011},
  organization={IEEE}
}

@article{Stackelberg_game,
  title={Dynamic pricing and distributed energy management for demand response},
  author={Jia, Liyan and Tong, Lang},
  journal={IEEE Transactions on Smart Grid},
  volume={7},
  number={2},
  pages={1128--1136},
  year={2016},
  publisher={IEEE}
}

@article{DR_qiadratic_review,
  title={Optimisation of demand response in electric power systems, a review},
  author={Jordehi, A Rezaee},
  journal={Renewable and sustainable energy reviews},
  volume={103},
  pages={308--319},
  year={2019},
  publisher={Elsevier}
}

@article{DR_qiadratic_reason,
  title={A survey on demand response in smart grids: Mathematical models and approaches},
  author={Deng, Ruilong and Yang, Zaiyue and Chow, Mo-Yuen and Chen, Jiming},
  journal={IEEE Transactions on Industrial Informatics},
  volume={11},
  number={3},
  pages={570--582},
  year={2015},
  publisher={IEEE}
}

@article{optimal_response,
  title={Optimal demand response: Problem formulation and deterministic case},
  author={Chen, Lijun and Li, Na and Jiang, Libin and Low, Steven H},
  journal={Control and optimization methods for electric smart grids},
  pages={63--85},
  year={2012},
  publisher={Springer}
}

@article{learning_data,
  title={Data-driven targeting of customers for demand response},
  author={Kwac, Jungsuk and Rajagopal, Ram},
  journal={IEEE Transactions on Smart Grid},
  volume={7},
  number={5},
  pages={2199--2207},
  year={2015},
  publisher={IEEE}
}

@article{learning_model_free,
  title={Data-driven modelling of energy demand response behaviour based on a large-scale residential trial},
  author={Antonopoulos, Ioannis and Robu, Valentin and Couraud, Benoit and Flynn, David},
  journal={Energy and AI},
  volume={4},
  pages={100071},
  year={2021},
  publisher={Elsevier}
}

@article{feedback,
  title={Reinforcement learning for demand response: A review of algorithms and modeling techniques},
  author={V{\'a}zquez-Canteli, Jos{\'e} R and Nagy, Zolt{\'a}n},
  journal={Applied energy},
  volume={235},
  pages={1072--1089},
  year={2019},
  publisher={Elsevier}
}

@article{learning_parameter,
  title={A distributed online pricing strategy for demand response programs},
  author={Li, Pan and Wang, Hao and Zhang, Baosen},
  journal={IEEE Transactions on Smart Grid},
  volume={10},
  number={1},
  pages={350--360},
  year={2017},
  publisher={IEEE}
}

@article{robust,
  title={Energy pricing and dispatch for smart grid retailers under demand response and market price uncertainty},
  author={Wei, Wei and Liu, Feng and Mei, Shengwei},
  journal={IEEE transactions on smart grid},
  volume={6},
  number={3},
  pages={1364--1374},
  year={2014},
  publisher={IEEE}
}

@inproceedings{cvar2,
  title={Optimal pricing for residential demand response: A stochastic optimization approach},
  author={Jia, Liyan and Tong, Lang},
  booktitle={2012 50th Annual Allerton Conference on Communication, Control, and Computing (Allerton)},
  pages={1879--1884},
  year={2012},
  organization={IEEE}
}

@article{intro_1,
  title={Toward a consumer-centric grid: A behavioral perspective},
  author={Saad, Walid and Glass, Arnold L and Mandayam, Narayan B and Poor, H Vincent},
  journal={Proceedings of the IEEE},
  volume={104},
  number={4},
  pages={865--882},
  year={2016},
  publisher={IEEE}
}

@article{doe_document,
  title={Benefits of demand response in electricity markets and recommendations for achieving them},
  author={{US Dept. Energy}},
  journal={US Dept. Energy, Washington, DC, USA, Tech. Rep},
  year={2006},
  publisher={Citeseer}
}

@inproceedings{xu2020operational,
  title={Operational valuation of energy storage under multi-stage price uncertainties},
  author={Xu, Bolun and Korp{\aa}s, Magnus and Botterud, Audun},
  booktitle={2020 59th IEEE Conference on Decision and Control (CDC)},
  pages={55--60},
  year={2020},
  organization={IEEE}
}

@ARTICLE{battery,
  author={Zheng, Ningkun and Jaworski, Joshua and Xu, Bolun},
  journal={IEEE Transactions on Power Systems}, 
  title={Arbitraging Variable Efficiency Energy Storage Using Analytical Stochastic Dynamic Programming}, 
  year={2022},
  volume={37},
  number={6},
  pages={4785-4795},
  doi={10.1109/TPWRS.2022.3154353}}

@article{thermal_nonlinear,
  title={Model predictive control for thermal energy storage and thermal comfort optimization of building demand response in smart grids},
  author={Tang, Rui and Wang, Shengwei},
  journal={Applied Energy},
  volume={242},
  pages={873--882},
  year={2019},
  publisher={Elsevier}
}

@misc{prudent_earliet,
  title={Precautionary Saving in the Small and in the Large},
  author={Kimball, Miles S},
  year={1989},
  booktitle={{National Bureau of Economic Research Cambridge, Mass., USA}}
}

@incollection{pratt_risk,
  title={Risk aversion in the small and in the large},
  author={Pratt, John W},
  booktitle={Uncertainty in economics},
  pages={59--79},
  year={1978},
  publisher={Elsevier}
}

@article{two_period_1,
  title={Optimal prevention and prudence in a two-period model},
  author={Menegatti, Mario},
  journal={Mathematical Social Sciences},
  volume={58},
  number={3},
  pages={393--397},
  year={2009},
  publisher={Elsevier}
}

@article{ambiguity,
  title={The impact of ambiguity and prudence on prevention decisions},
  author={Berger, Lo{\"\i}c},
  journal={Theory and Decision},
  volume={80},
  pages={389--409},
  year={2016},
  publisher={Springer}
}

@article{skewness_left,
  title={Putting risk in its proper place},
  author={Eeckhoudt, Louis and Schlesinger, Harris},
  journal={American Economic Review},
  volume={96},
  number={1},
  pages={280--289},
  year={2006},
  publisher={American Economic Association}
}

@article{skewness_1,
  title={Joint measurement of risk aversion, prudence, and temperance},
  author={Ebert, Sebastian and Wiesen, Daniel},
  journal={Journal of Risk and Uncertainty},
  volume={48},
  pages={231--252},
  year={2014},
  publisher={Springer}
}

@article{experiment_2,
  title={Testing for prudence and skewness seeking},
  author={Ebert, Sebastian and Wiesen, Daniel},
  journal={Management Science},
  volume={57},
  number={7},
  pages={1334--1349},
  year={2011},
  publisher={INFORMS}
}

@article{time,
  title={Decision making when things are only a matter of time},
  author={Ebert, Sebastian},
  journal={Operations Research},
  volume={68},
  number={5},
  pages={1564--1575},
  year={2020},
  publisher={INFORMS}
}

@article{technology_risk,
  title={Should we do more when we know less? The effect of technology risk on optimal effort},
  author={Li, Lu and Peter, Richard},
  journal={Journal of Risk and Insurance},
  volume={88},
  number={3},
  pages={695--725},
  year={2021},
  publisher={Wiley Online Library}
}

@article{risk_averse,
  title={New results on the relationship among risk aversion, prudence and temperance},
  author={Menegatti, Mario},
  journal={European Journal of Operational Research},
  volume={232},
  number={3},
  pages={613--617},
  year={2014},
  publisher={Elsevier}
}

@article{high_order,
  title={New results on high-order risk changes},
  author={Menegatti, Mario},
  journal={European Journal of Operational Research},
  volume={243},
  number={2},
  pages={678--681},
  year={2015},
  publisher={Elsevier}
}

@article{loss_averse,
  title={Loss-averse preferences and portfolio choices: An extension},
  author={Eeckhoudt, Louis and Fiori, Anna Maria and Gianin, Emanuela Rosazza},
  journal={European Journal of Operational Research},
  volume={249},
  number={1},
  pages={224--230},
  year={2016},
  publisher={Elsevier}
}

@article{change_more,
  title={Sometimes more, sometimes less: Prudence and the diversification of risky insurance coverage},
  author={Reichel, Lukas and Schmeiser, Hato and Schreiber, Florian},
  journal={European Journal of Operational Research},
  volume={292},
  number={2},
  pages={770--783},
  year={2021},
  publisher={Elsevier}
}

@misc{real_time_price,
    author = {{Ameren}},
   title = {Power Smart Pricing Plan},
note = {Retrieved Apr 27},
   howpublished = {\url{https://www.ameren.com/illinois/account/ customer-service/bill/power-smart-pricing/}},
    year = {2025}
}

@misc{nyt,
  title = {His Lights Stayed on During Texas’ Storm. Now He Owes \$16,752.},
  author={Giulia, McDonnell Nieto del Rio and Nicholas, Bogel-Burroughs and Ivan, Penn},
  note = {\url{https://www.nytimes.com/2021/02/20/us/texas-storm-electric-bills.html\#commentsContainer}},
howpublished = {New York Times},
year = 2021,
}

@misc{TX_realtime,
   title = {Real-Time Market},
note = {Retrieved Apr 27},
    author = {{Ercot}},
   howpublished = {\url{https://www.ercot.com/mktinfo/rtm/}},
year = {2025}
}

@misc{PJM_real,
   title = {Energy Market},
note = {Retrieved Apr 27},
    author = {{PJM}},
   howpublished = {\url{https://www.pjm.com/markets-and-operations/energy}},
year = {2025}
}

@misc{Con_ed_ny,
   title = {Con Edison Time-of-Use Rates},
note = {Retrieved Apr 27},
 author = {ConEd},
   howpublished = {\url{https://www.coned.com/en/accounts-billing/your-bill/time-of-use}},
year = {2025}
}

@misc{SRP,
   title = {Salt River Project (SRP) Time-of-Use Price Plan™ (TOU)},
note = {Retrieved Apr 27},
author = {SRP},
   howpublished = {\url{https://www.srpnet.com/price-plans/residential-electric/time-of-use}},
year = {2025}
}

@misc{PGE,
   title = {Find your best rate plan},
author ={PG\&E},
note = {Retrieved Apr 27},
   howpublished = {\url{https://www.pge.com/en/account/rate-plans/find-your-best-rate-plan.html}},
year = {2025}
}

@misc{OGE,
   title = {SmartHours},
author ={OG\&E},
note = {Retrieved Apr 27},
   howpublished = {\url{https://www.oge.com/wps/portal/ord/residential/pricing-options/smart-hours}},
year = {2025}
}

@article{tariff_dynamic1,
  title={The economics of fixed cost recovery by utilities},
  author={Borenstein, Severin},
  journal={The Electricity Journal},
  volume={29},
  number={7},
  pages={5--12},
  year={2016},
  publisher={Elsevier}
}

@book{convex,
  title={Convex optimization},
  author={Boyd, Stephen P and Vandenberghe, Lieven},
  year={2004},
  publisher={Cambridge university press}
}

\newpage

\appendix
\renewcommand{\thesection}{Appendix \Alph{section}}

\section{Proof of Theorem \ref{price_invariant_theorem}}
\emph{Overview of the proof:} 
\todo{We first prove the one-step distribution-insensitive condition and derive a series of lemmas that will work for the subsequent analysis. We derive the formula of the derivative of action-value function $q_{t}(x_t|\lambda_{t+1})$ and show the relationship between $q_{t}(x_t|\lambda_{t+1}), p_{t+1}$ and future price $\lambda_{t+1}$ (Lemma \ref{q_t-1_lambda}).} We then show that the state-value function derivative $v$ is a linear combination of the state cost function derivative $c$ (Lemma \ref{linear_combination_value_function}). \todo{Based on this, we show that the linear relationship between $q_{t}(x_t|\lambda_{t+1})$ and $\lambda_{t+1}$ is necessary and sufficient for the one-step distribution-insensitive condition. Finally, we extend the one-step property to all future stages and complete the proof of the Theorem.}

\begin{lemma}\emph{Relationship between $p_{t+1},q_{t}(x_t|\lambda_{t+1})$ and $\lambda_{t+1}$}. \label{q_t-1_lambda}
Consider the demand model describe in (\ref{eq1}), for all $t\in [0,T-1]$, 
    \begin{subequations}
    \begin{align}
    \frac{\partial p_{t+1}}{\partial \lambda_{t+1}} = 
    -\frac{1}{\dot{g}_{t+1}(p_{t+1}) + \dot{h}_{t+1}(x_{t+1})},\\
    \frac{\partial q_{t}(x_t|\lambda_{t+1})}{\partial \lambda_{t+1}} = 
    A\frac{\partial h_{t+1}(x_{t+1})}{\partial \lambda_{t+1}} = 
    \frac{-A \dot{h}_{t+1}(x_{t+1})}{\dot{g}_{t+1}(p_{t+1}) + \dot{h}_{t+1}(x_{t+1}) }. \label{h_t/lamba_t}
    \end{align}
    \end{subequations}
\end{lemma}

\begin{proof}
\todo{We first link the current action-value function to the future state-related cost function, i.e., proving $q_{t}(x_{t}|\lambda_{t+1})=A h_{t+1}(x_{t+1}).$ }

We apply the optimality condition to $Q_{t}(x_t|\lambda_{t+1})$ to find the minimized $p_{t+1}$,
    \begin{align}
    \frac{\partial Q_{t}(x_t|\lambda_{t+1})}{\partial p_{t+1}} &= \lambda_{t+1} + \frac{\partial C_{t+1}(x_{t+1})}{\partial x_{t+1}}
    \frac{\partial x_{t+1}}{\partial p_{t+1}} 
    +\frac{\partial V_{t+1}(x_{t+1})}{\partial x_{t+1}} 
    \frac{\partial x_{t+1}}{\partial p_{t+1}} + 
    \frac{\partial G_{t+1}(p_{t+1})}{\partial p_{t+1}}=0, \label{q/p}
\end{align}
and according to (\ref{cons}), 
    $\partial x_{t+1} / \partial p_{t+1} = 1$.
Then the optimality condition is equivalent to 
\begin{subequations}
\begin{align}
     &\lambda_{t+1} + g_{t+1}(p_{t+1}) + h_{t+1}(x_{t+1}) = 0, \label{p*} \\
     &x_{t+1} = A x_{t} + p_{t+1}.
\end{align}
\end{subequations}

With the optimal $p_{t+1}$, we take derivative for $Q_{t}(x_t|\lambda_{t+1})$ with regards to $x_{t}$,
\begin{subequations}
\begin{align}
    \frac{\partial Q_{t}(x_t|\lambda_{t+1})}{\partial x_{t}} = 
    \lambda_{t+1} \frac{\partial p_{t+1}}{\partial x_t} + 
    \frac{\partial C_{t+1}(x_{t+1})}{\partial x_{t+1}}
    \frac{\partial x_{t+1}}{\partial x_t} 
    +\frac{\partial V_{t+1}(x_{t+1})}{\partial x_{t+1}} 
    \frac{\partial x_{t+1}}{\partial x_{t}} + 
    g_{t+1}(p_{t+1})
    \frac{\partial p_{t+1}}{\partial x_t}, \label{q/x}  
\end{align}
and according to (\ref{cons}),
\begin{align}
    \frac{\partial x_{t+1}}{\partial x_{t}} = 
    A + \frac{\partial p_{t+1}}{\partial x_{t}}. \label{xt/xt-1}
\end{align}
Thus, combine with (\ref{p*}) and (\ref{xt/xt-1}),
\begin{align}
    q_{t}(x_t|\lambda_{t+1}) &= \lambda_{t+1} \frac{\partial p_{t+1}}{\partial x_{t}}+ 
    h_{t+1}(x_{t+1}) \frac{\partial x_{t+1}}{\partial x_{t}}+ 
    g_{t+1}(p_{t+1}) \frac{\partial p_{t+1}}{\partial x_{t}} \nonumber \\
    &= \lambda_{t+1} \frac{\partial p_{t+1}}{\partial x_{t}}+ 
    h_{t+1}(x_{t+1}) (A +\frac{\partial p_{t+1}}{\partial x_{t}})+ 
    g_{t+1}(p_{t+1}) \frac{\partial p_{t+1}}{\partial x_{t}} \nonumber \\
    &= \lambda_{t+1} \frac{\partial p_{t+1}}{\partial x_{t}}+ 
    h_{t+1}(x_{t+1}) (A + \frac{\partial p_{t+1}}{\partial x_{t}})  - (\lambda_{t+1} + h_{t+1}(x_{t+1})) \frac{\partial p_{t+1}}{\partial x_{t}} 
    = A h_{t+1}(x_{t+1}), \label{clean_q/x}
\end{align}
\end{subequations}
where the optimal $x_{t+1}$ is a function of $\lambda_{t+1}$.

\todo{With the relationship, we can analyze the relationship between the current action-value function and the future price. According to the optimality conditions for $p_{t+1}$ as described in (\ref{p*}),}

    
\begin{subequations}
\begin{align}
    &\frac{\partial (g_{t+1}(p_{t+1})+h_{t+1}(x_{t+1}))}{\partial \lambda_{t+1}} = 
    - \frac{\partial \lambda_{t+1}}{\partial \lambda_{t+1}},  \\
    &\frac{\partial g_{t+1}(p_{t+1})}{\partial p_{t+1}} 
    \frac{\partial p_{t+1}}{\partial \lambda_{t+1}} + 
    \frac{\partial h_{t+1}(x_{t+1})}{\partial x_{t+1}} 
    \frac{\partial x_{t+1}}{\partial p_{t+1}}
    \frac{\partial p_{t+1}}{\partial \lambda_{t+1}} = -1,  \\
    &\dot{g}_{t+1}(p_{t+1}) \frac{\partial p_{t+1}}{\partial \lambda_{t+1}} + 
    \dot{h}_{t+1}(x_{t+1}) \frac{\partial p_{t+1}}{\partial \lambda_{t+1}} =  -1, \\
    &\frac{\partial p_{t+1}}{\partial \lambda_{t+1}} = 
    -\frac{1}{\dot{g}_{t+1}(p_{t+1}) + \dot{h}_{t+1}(x_{t+1})}.      \label{pt*/lambat}
\end{align}
\end{subequations}

  Then, we take derivative of $q_{t}(x_t|\lambda_{t+1})$ with regards to $\lambda_{t+1}$,
        \begin{align}
       \frac{\partial q_{t}(x_t|\lambda_{t+1})}{\partial \lambda_{t+1}} &= 
        \frac{\partial A h_{t+1}(x_{t+1})}{\partial \lambda_{t+1}} = A
    \frac{\partial h_{t+1}(x_{t+1})}{\partial x_{t+1}} 
    \frac{\partial x_{t+1}}{\partial p_{t+1}}
    \frac{\partial p_{t+1}}{\partial \lambda_{t+1}}
     = \frac{-A\dot{h}_{t+1}(x_{t+1})}{\dot{g}_{t+1}(p_{t+1}) + \dot{h}_{t+1}(x_{t+1}) }.
        \end{align}
    This completes the proof of the Lemma.
\end{proof}

Involving the relationship between $h$ and $x$, Lemma \ref{q_t-1_lambda} provide an exact relationship between $p_{t+1}, q_{t}(x_t|\lambda_{t+1})$ and $\lambda_{t+1}$. However, the conditions described in Theorem \ref{price_invariant_theorem} only state the properties of the $C,G$ rather than $h$ or $v$ with regards to $x$. We then introduce the following Lemma to connect $v$ with $c$. 

\begin{lemma}\emph{State-value function derivative}.\label{linear_combination_value_function}
    The state-value function derivative $v_t(x_t)$ of the demand model described in (\ref{eq1}) satisfy the following for all $t\in [1,T]$: 
    \begin{align}
        v_t(x_t) = \sum_{\tau=t+1}^{T} a_{\tau} c_{\tau}(x_{\tau}) +a_{\mathrm{v}} v_T(x_T), \label{v_t_linear_comb}
    \end{align}
    where $a_{\tau}, a_{\mathrm{v}}$ is the parameters for the linear combination.
\end{lemma} \label{linear_combination_v_t_prop}

\begin{proof}[Proof]
    Taking derivative of (\ref{3b}) with regards to $x_t$,
    \begin{align}
        v_{t}(x_{t})=\mathbb{E}_{\Lambda_{t+1}}[q_{t}(x_{t}|\lambda_{t+1})],
    \end{align}
    and according to (\ref{clean_q/x}) from Lemma \ref{q_t-1_lambda},
    \begin{subequations}   
    \begin{align}
        q_{T-1}(x_{T-1}|\lambda_{T}) = Ah_T(x_T)= A(v_T(x_T) + c_T(x_T)).
    \end{align}
    Recursively,
    \begin{align}
        q_{T-2}(x_{T-2}|\lambda_{T-1}) = A (v_{T-1}(x_{T-1}) + c_{T-1}(x_{T-1})
        =A(\mathbb{E}[A(v_T(x_T) + c_T(x_T))] + c_{T-1}(x_{T-1})).  
    \end{align}   
     \end{subequations}
     
    Then, 
    \begin{align}
        v_t(x_t) = \mathbb{E}\bigg[a_{t+1} c_{t+1}(x_{t+1})
        + \mathbb{E}\big[ a_{t+2} c_{t+2}(x_{t+2})+ \dots
        +\mathbb{E}[a_{T-1}c_{T-1}(x_{T-1}) + \mathbb{E}[a_{T} c_{T}(x_{T})+ a_{\mathrm{v}}v_T(x_T)]] \big] \bigg].      
    \end{align}
    where $a_{\tau},\tau \in [t+1,T]$ and $a_{\mathrm{v}}$ is the constant parameters.
    
    As expectation is a linear transformation, we can write $v_t(x_t)$ as (\ref{v_t_linear_comb}) and complete the proof.
\end{proof}

Lemma \ref{linear_combination_value_function} shows that the state-value function derivative is a linear combination of the state cost function derivative $c$ and the end state-value function derivative $v_T$. As the model definition mentions, we reasonably set the end state-value function to zero and finish the proof. 

\begin{proof}[Proof of Theorem \ref{price_invariant_theorem}]
\todo{We first prove the one-step distribution-insensitive condition.} According to Lemma \ref{linear_combination_value_function}, we rewrite the expression of $h_{t+1} = c_{t+1} +v_{t+1}$ as $\sum_{\tau=t+1}^T a_{\tau}c_\tau$ with an zero end state-value function. According to the conditions described in Theorem \ref{price_invariant_theorem}, if the state cost function $C$ and action cost function $G$ are quadratic, then $h$ is a linear function with regard to $x$. Combining with Lemma \ref{q_t-1_lambda}, denote $M_1, M_2$ as a constant, $M_1,M_2\in 
\mathbb{R}$, $(\partial p_{t+1})/(\partial \lambda_{t+1}) = M_1 $ and $(\partial q_t(x_t|\lambda_{t+1}))/(\partial \lambda_{t+1}) =(\partial h_{t+1}(x_{t+1}))/(\partial \lambda_{t+1}) = M_2$. This means a linear relationship between both $p_{t+1}, \lambda_{t+1}$ and $q_{t}(x_t|\lambda_{t+1}), \lambda_{t+1}$. 

We first take an example with a two-point price distribution given $w \in[0,1]$,
\begin{subequations}
    \begin{align}
        \gamma_{t+1} &= \mathbb{E}[\lambda_{t+1}] - w\delta_{t+1} \text{ with a probability of $1-w$},\\
        \pi_{t+1} &= \mathbb{E}[\lambda_{t+1}] + (1-w)\delta_{t+1} \text{ with a probability of $w$},
    \end{align}     
\end{subequations}
    which ensures the expectation of this two-point price distribution is always $\mathbb{E}[\lambda_{t+1}]$ as $(1-w)\gamma_{t+1} + w\pi_{t+1} = \mathbb{E}[\lambda_{t+1}]$, while with any given $w$, the variance of the distribution scales monotonically with $\delta_{t+1} = \pi_{t+1}-\gamma_{t+1}$.
    
Then, we take $\lambda_{t+1},\gamma_{t+1}$ into the $q_{t}(x_t|\lambda_{t+1})$ expression and explicitly get, 
\begin{subequations} \label{distribution_insensitive_conditions}
\begin{align}
    (1-w)q_{t}(x_{t}|\gamma_{t+1})
    &= A (1-w)h_{t+1}(x_{t}+p_{t+1}(\mathbb{E}[\lambda_{t+1}] - w\delta_{t+1})) \nonumber \\ 
    &= A (1-w)h_{t+1}(x_{t}+p_{t+1}(\mathbb{E}[\lambda_{t+1}])) - A w(1-w)h_{t+1}(p_{t+1}(\delta_{t+1})),  \\
    wq_{t}(x_{t}|\pi_{t+1}) 
    &= A w h_{t+1}(x_{t}+p_{t+1}(\mathbb{E}[\lambda_{t+1}] + (1-w)\delta_{t+1}) ) \nonumber \\ 
    &= A w h_{t+1}(x_{t}+p_{t+1}(\mathbb{E}[\lambda_{t+1}])) + A w(1-w)h_{t+1}(p_{t+1}(\delta_{t+1})).
\end{align}
\end{subequations}

Thus,
    \begin{align}
        w q_{t}(x_{t}|\pi_{t+1}) +(1-w) q_{t}(x_{t}|\gamma_{t+1}) 
        = A h_{{t+1}}(x_{t}+p_{t+1}(\mathbb{E}[\lambda_{t+1}]) ) = q_{t}(x_{t}|\mathbb{E}[\lambda_{t+1}]). \label{final_insensitive}
    \end{align}

Now we consider a general distribution $\Lambda_{t+1}$, and due to the linear relationship between $q_t(x_t|\lambda_{t+1})$ and $\lambda_{t+1}$, we can write (\ref{final_insensitive}) as follows,
\begin{align}
    \mathbb{E}_{\Lambda_{t+1}}[q_t(x_t|\lambda_{t+1})]
    = q_t(x_t|\mathbb{E}_{\Lambda_{t+1}}[\lambda_{t+1}]),\label{final_insensitive1}
\end{align}
which shows that the derivative of the action-value function is distribution-insensitive to the price distribution, but only to the expectation. 

Then we prove (\ref{final_insensitive1}) is a sufficient and necessary condition for the distribution-insensitive demand.
In terms of the sufficient condition, given (\ref{final_insensitive1}), when the expectation of price distribution is fixed, $q_{t}$ is insensitive to the price distribution. This means the state-value function derivative $v_t$ is the same when taking different distributions at stage $t+1$, indicating the same demand action and state at stage $t$.
    
For the necessary condition, given demand at stage $t$ is insensitive to price distribution at stage $t+1$, then $v_t$ and $q_{t}$ are only affected by the expectation instead of the price distributions. This means the expectation can be moved to the price variable ($\mathbb{E}_{\Lambda_{t+1}}[\lambda_{t+1}]$) directly, corresponding to a linear $q_{t}$ function with regards to price $\lambda_{t+1}$, as (\ref{final_insensitive1}) shows.

\todo{Finally, we extrapolate the one-step property to all future stages. We have shown that consumption at stage $t$ is independent of the price distribution at stage $t+1$, and the independence also holds for consumption at stage $t+1$ and price distribution at stage $t+2$. Following the state transition process, the consumption at stage $t$ is independent of the price distribution at stage $t+2$. Thus, recursively, consumption at stage $t$ is independent of all future price distributions. This completes the proof of Theorem \ref{price_invariant_theorem}.}
\end{proof}

\section{Proof of Corollary \ref{corollary_1}}
\begin{proof}
We first show the extra case. As the system is linear and cost function formulation is symmetrical, with a prior state of zero ($x_{t-1} = 0$) and symmetry price distribution with a mean of zero ($\mathbb{E}_{\Lambda_{t+1}}[\lambda_{t+1}] =0$), demands action should be zero. This shows the distribution-insensitive demand. 

We then prove the other parts of the Corollary from three perspectives. The first is the superquadratic state cost function and quadratic action cost function. 

From Lemma \ref{q_t-1_lambda}, denote $\alpha_1$ as a constant, 
    \begin{align}
        \frac{\partial h_{t+1}(x_{t+1})}{\partial \lambda_{t+1}} = 
        - \frac{ \dot{h}_{t+1}(x_{t+1})}{\alpha_1 + \dot{h}_{t+1}(x_{t+1})}, 
    \end{align}
    
    Thus, $\partial h_{t+1}(x_{t+1})/\partial \lambda_{t+1}$ is not a constant, and $h_{t+1}(x_{t+1})$ is not a linear function with regard to $\lambda_{t+1}$. According to the proof of Theorem \ref{price_invariant_theorem}, the distribution-insensitive condition doesn't hold, and the demand model is distribution-sensitive.

    Then, we prove the superquadratic action cost function and quadratic state cost function. Also, from Lemma \ref{q_t-1_lambda}, denote $\alpha_2$ as a constant,
    \begin{align}
        \frac{\partial h_{t+1}(x_{t+1})}{\partial \lambda_{t+1}} = 
        - \frac{\alpha_2}{\dot{g}_{t+1}(p_{t+1})}, 
    \end{align}
    
    Following the same analysis, we show that the demand model becomes distribution-sensitive. These two cases also show that demand is distribution-sensitive with superquadratic state and action cost function, and prove this Corollary.
\end{proof}
    
\section{Proof of Theorem \ref{aggressive_theorem}}
\emph{Overview of the proof.} 
In our proof, we set the event to happen at stage $t+1$ and use the price model for stage $t+1$ as described in (\ref{price_model}). By model definition, we reasonably set the end state-value function $V_{T}$ to zero. 
Here, we only consider the decision behavior before the event at the stage $t+1$ instead of the situation after the event happens, which is independent of the $V_{t+1}(x_{t+1})$; thus, we set $V_{t+1}(x_{t+1}) = 0$. 
We specify the proof process as follows:
\begin{itemize}
    \item We first set $V_{t+1}(x_{t+1}) = 0$ and write the extended form for the first stage optimization problem with price $(\lambda_{\tau} =0, \tau \in [1,t],\gamma_{t+1},\pi_{t+1})$, and take the optimality conditions with regards to actions. Then we write the optimality conditions as a function of the future prices to analyze the optimal solution (Lemma \ref{aggressive_theorem_lamma1}). 
    \item We analyze and show the function relationship between price $\gamma_{t+1},\pi_{t+1}$ and optimal state value $x_{t+1}$, and prove its monotonicity, symmetry, and concavity (Lemma \ref{aggressive_theorem_lamma2}).
    \item By using the function property obtained from Lemma \ref{aggressive_theorem_lamma2}, we show the $t$th stage reverse state-value function derivative $-v_t$ is positive under the optimality condition from Lemma \ref{aggressive_theorem_lamma1} and is greater than the value under the price expectation condition, and the same is true for the state value $x_t$ (Lemma \ref{aggressive_theorem_lamma3}). 
    \item We further analyze the sensitivity of the reverse state-value function derivative and show that the state value $x_t$ strictly increases with future price $\pi_{t+1}$ (Lemma \ref{aggressive_theorem_lamma4}). 
\end{itemize}
Combining these analyses shows that the state value $x_t$ is positive, distribution sensitive, and strictly increases with a more right-skewed price distribution.

\begin{lemma}\emph{Demand's optimality conditions.} \label{aggressive_theorem_lamma1}
    Taking the price $\gamma_{t+1},\pi_{t+1}$ from distribution $\Lambda_{t+1}$ and demand model described in Theorem \ref{aggressive_theorem}, assuming $V_{t+1}(x_{t+1})=0$, the optimality conditions for the first stage action of the demand model is:
    \begin{subequations} \label{optimality_conditions} 
    \begin{align}
            &a_{\mathrm{p}} p_1 + \sum_{\tau=1}^{t} c_{\tau}(x_{\tau}) -\mathbb{E}[\lambda_{t+1}] + a_{\mathrm{p}} Ax_{t} 
            - a_{\mathrm{p}} [w_{\Lambda_{t+1}} \phi(-\pi_{t+1} + a_{\mathrm{p}}A x_{t} ) 
            + (1-w_{\Lambda_{t+1}}) \phi(\gamma_{t+1} + a_{\mathrm{p}} Ax_{t} )] = 0, 
    \end{align}
    where $\phi^{-1}(x_{t+1}) = a_{\mathrm{p}} x_{t+1} + c_{t+1}(x_{t+1})$, describing the function relationship between optimal state $x_{t+1}$ and price $\lambda_{t+1}$.  
    \end{subequations}  
\end{lemma}

\begin{proof}
    According to the demand model define in Theorem \ref{aggressive_theorem} and setting $V_{t+1}(x_{t+1})=0$, 
        \begin{align}
            &Q_t(x_t|\lambda_{t+1}) = \min_{p_{t+1}}\lambda_{t+1} p_{t+1} + \frac{a_{\mathrm{p}}} {2}p_{t+1}^{2} + C_{t+1}(x_{t+1}), \nonumber \\
            &\text{s.t. } x_{t+1} = Ax_{t} + p_{t+1}.
        \end{align}
        
        Taking the price scenarios ($\gamma_{t+1},\pi_{t+1}$) inside, we get the state-value function $V_t(x_t)$, (For brevity, we omit the subscript of $w_{\Lambda_{t+1}}$ and express as $w$ without confusion),
        \begin{subequations}
        \begin{align}
            V_t(x_t) = \mathbb{E}_{\Lambda_{t+1}}[Q_t(x_t|\lambda_{t+1})] 
            &= \min_{p_{\mathrm{\pi},t+1}, p_{\mathrm{\gamma},t+1}} 
            w[\pi_{t+1} p_{\mathrm{\pi},t+1} + \frac{a_{\mathrm{p}}} {2}p_{\mathrm{\pi},t+1}^2 + C_{t+1}(x_{\mathrm{\pi},t+1})] \nonumber \\ 
            &+ (1-w) [-\gamma_{t+1} p_{\mathrm{\gamma},t+1} + \frac{a_{\mathrm{p}}} {2}p_{\mathrm{\gamma},t+1}^2 + C_{t+1}(x_{\mathrm{\gamma},t+1})],    \\
            \text{s.t. } x_{\mathrm{\pi},t+1} &= A x_{t} + p_{\mathrm{\pi},t+1}, \\
            x_{\mathrm{\gamma},t+1} &= A x_{t} + p_{\mathrm{\gamma},t+1}, 
        \end{align}
        \end{subequations}
        where $p_{\mathrm{\pi},t+1}, x_{\mathrm{\pi},t+1}$ is the actions and state in stage $t+1$ under price scenario $\pi_{t+1}$; $p_{\mathrm{\gamma},t+1}, x_{\mathrm{\gamma},t+1}$ is the actions and state in stage $t+1$ under price scenario $\gamma_{t+1}$.
        
        Thus, we express the optimization problem for the demand recursively, starting from the first stage as 
        \begin{subequations}
        \begin{align}
            Q_{0}(x_{0}|\lambda_{\tau,\tau \in[1,t]}) = \min_{p_{\mathrm{\gamma},t+1},p_{\mathrm{\pi},t+1},p_{\tau,\tau \in [1,t]}} 
            &\sum_{\tau=1}^{t} [\lambda_{\tau} p_{\tau} + \frac{a_{\mathrm{p}}}{2}p_{\tau}^2 + C_{\tau}(x_{\tau})]  \nonumber \\
            +w[\pi_{t+1} p_{\mathrm{\pi},t+1} + \frac{a_{\mathrm{p}}} {2}p_{\mathrm{\pi},t+1}^2 + C_{t+1}(x_{\mathrm{\pi},t+1})] 
            &+(1-w) [-\gamma_{t+1} p_{\mathrm{\gamma},t+1} + \frac{a_{\mathrm{p}}} {2}p_{\mathrm{\gamma},t+1}^2 + C_{t+1}(x_{\mathrm{\gamma},t+1})],      \\
            \text{s.t. } 
            x_{\tau} &= A x_{\tau-1}+p_{\tau} ,\ \forall \tau \in[1,t] \label{constraint1} \\
            x_{\mathrm{\pi},t+1} &= A x_{t} + p_{\mathrm{\pi},t+1}, \\
            x_{\mathrm{\gamma},t+1} &= A x_{t} + p_{\mathrm{\gamma},t+1}. \label{constraint2}
        \end{align}
        \end{subequations}

        Then, we take the optimality conditions for the decision variables, and highlight the conditions of $p_1,p_{\mathrm{\gamma},t+1},p_{\mathrm{\pi},t+1}$ and take $\lambda_{\tau}=0,\tau \in[1,t]$ inside:
    \begin{subequations}
        \begin{align}
        &p_1: a_{\mathrm{p}} p_1 + \sum_{\tau=1}^{t}  c_{\tau}(x_{\tau}) +w c_{t+1}(x_{\mathrm{\pi},t+1}) +(1-w) c_{t+1}(x_{\mathrm{\gamma},t+1}) =0, \label{1} \\
        &p_{\mathrm{\pi},t+1}: \pi_{t+1} + a_{\mathrm{p}} p_{\mathrm{\pi},t+1} +  c_{t+1}(x_{\mathrm{\pi},t+1}) =0, \label{2}\\
        &p_{\mathrm{\gamma},t+1}: -\gamma_{t+1} + a_{\mathrm{p}} p_{\mathrm{\gamma},t+1} + c_{t+1}(x_{\mathrm{\gamma},t+1}) =0, \label{3}
    \end{align}
    among them, (\ref{constraint1})-(\ref{constraint2}) also stands with optimal $x$ and $p$.
    \end{subequations}
    
Then, by multiplying (\ref{2}) by $w$, multiplying (\ref{3}) by $1-w$, and taking this into (\ref{1}), we can rewrite the $t$th stage state-value function derivative in (\ref{1}),
\begin{subequations}
\begin{align}
    a_{\mathrm{p}} p_1 &+\sum_{\tau=1}^{t} c_{\tau}(x_{\tau}) 
    -w\pi_{t+1} - w a_{\mathrm{p}} p_{\mathrm{\pi},t+1} +(1-w)\gamma_{t+1} -(1-w)a_{\mathrm{p}} p_{\mathrm{\gamma},t+1} = 0,  
\end{align}
as $w\pi_{t+1} - (1-w)\gamma_{t+1} = \mathbb{E}[\lambda_{t+1}]$, 
\begin{align}
    a_{\mathrm{p}} p_1 + \sum_{\tau=1}^{t} &c_{\tau}(x_{\tau}) -\mathbb{E}[\lambda_{t+1}]  
    - a_{\mathrm{p}} [wp_{\mathrm{\pi},t+1} + (1-w) p_{\mathrm{\gamma},t+1}] = 0.  \label{update_1}
\end{align}
\end{subequations}

Then, to find the relationships between the state-value function derivative 
$v_t$ and price $\lambda,\gamma,\pi$, we replace the action $p_{\mathrm{\pi},t+1}$ and $p_{\mathrm{\gamma},t+1}$ in (\ref{update_1}) as the price $\pi_{t+1}$ and $\gamma_{t+1}$, respectively. We connect the price with optimal action/state according to (\ref{2}),
\begin{subequations}
    \begin{align}
       \pi_{t+1}  &= - a_{\mathrm{p}} p_{\mathrm{\pi},t+1} -  c_{t+1}(x_{\mathrm{\pi},t+1}) 
       = - a_{\mathrm{p}} x_{\mathrm{\pi},t+1} -  c_{t+1}(x_{\mathrm{\pi},t+1}) + a_{\mathrm{p}} A  x_t,  \\
        \pi_{t+1} &- a_{\mathrm{p}} A x_t  = -\phi^{-1} (x_{\mathrm{\pi},t+1}), \\
        x_{\mathrm{\pi},t+1} & = \phi(-\pi_{t+1}+ a_{\mathrm{p}} A x_t ),
    \end{align}
    where $\phi^{-1}(x_{t+1}) = a_{\mathrm{p}} x_{t+1} + c_{t+1}(x_{t+1})$.
    
    Applying the same steps to (\ref{3}), 
    \begin{align}
        \gamma_{t+1} + a_{\mathrm{p}}  A x_{t} 
        &= a_{\mathrm{p}} x_{\mathrm{\gamma},t+1} + c_{t+1}(x_{\mathrm{\gamma},t+1}) = \phi^{-1} (x_{\mathrm{\gamma},t+1}), \\
        x_{\mathrm{\gamma},t+1} &= \phi(\gamma_{t+1} + a_{\mathrm{p}} A x_{t} ).  \label{4}
    \end{align}
\end{subequations}

Then we rewrite (\ref{update_1}) as $x_{t+1}$ expression by adding $A x_t $ and replace $x_{\mathrm{\gamma},t+1},x_{\mathrm{\pi},t+1}$ to the function of price $\gamma_{t+1},\pi_{t+1}$.
\begin{subequations}
\begin{align}
    a_{\mathrm{p}} p_1 + \sum_{\tau=1}^{t} c_{\tau}(x_{\tau}) - \mathbb{E}[\lambda_{t+1}]  
    & - a_{\mathrm{p}} [w x_{\mathrm{\pi},t+1} + (1-w) x_{\mathrm{\gamma},t+1} - A x_{t} ] = 0,    \\
     a_{\mathrm{p}} p_1 + \sum_{\tau=1}^{t} c_{\tau}(x_{\tau}) - \mathbb{E}[\lambda_{t+1}] +  a_{\mathrm{p}} Ax_{t}
    & - a_{\mathrm{p}} [w\phi(-\pi_{t+1} + a_{\mathrm{p}} Ax_{t} ) 
    + (1-w) \phi(\gamma_{t+1} + a_{\mathrm{p}} Ax_{t})] = 0. 
\end{align}
\end{subequations}

Thus, we finish the proof of Lemma \ref{aggressive_theorem_lamma1}.
\end{proof}
\begin{remark}\emph{Optimality conditions with price expectation.}
\label{remark_proof}
    Follow Lemma \ref{aggressive_theorem_lamma1}, when taking the expectation of price distribution $\Lambda_{t+1}$, which is $\mathbb{E}[\lambda_{t+1}]$, the same optimality conditions of the demand model in Theorem \ref{aggressive_theorem} is
    \begin{align}
        a_{\mathrm{p}} p_1 + \sum_{\tau=1}^{t} c_{\tau}(x_{\tau}) - \mathbb{E}[\lambda_{t+1}] +  a_{\mathrm{p}} Ax_{t}
     - a_{\mathrm{p}} \phi(-\mathbb{E}[\lambda_{t+1}] + a_{\mathrm{p}} Ax_{t}) = 0. 
    \end{align}
\end{remark}

An intuitive explanation of this Lemma is to replace the state-value function derivative from state expression $x_{t+1}$ to price expression $(\mathbb{E}[\lambda_{t+1}],\gamma_{t+1},\pi_{t+1})$. The key to analyzing the optimality conditions is to show the function property of $\phi(\lambda)$, which motivates us to provide the following Lemma. 

\begin{lemma}\emph{Property of $\phi(\lambda)$.} \label{aggressive_theorem_lamma2}
    The inverse function $\phi(\lambda)$ from $\phi^{-1}(x)$ described in Lemma \ref{aggressive_theorem_lamma1} is strictly increasing, symmetric about the origin, strictly concave when $\lambda \geq 0$, and strictly convex when $\lambda < 0$ and $\phi(0) = 0$.
\end{lemma}

\begin{proof}
By definition, the state cost function $C(x)$ is continuously differentiable, superquadratic, convex, and symmetrical with the y-axis. This means the first-order derivative $c(x)$ is strictly convex when $x\geq0$ and strictly concave when $x < 0$. We also have $c(x)=-c(-x)$ as $c(x)$ is symmetric about the origin.

Then we analyze the property of $\phi^{-1}(x) = c(x)+a_{\mathrm{p}} x$. 
\begin{itemize}
    \item For convexity, the affine function $a_{\mathrm{p}} x$ doesn't affect the function convexity, so the function convexity is determined by $c(x)$. This means $\phi^{-1}(x)$ is strictly convex when $x\geq0$ and strictly concave when $x<0$.
    \item Because the $c(x)$ is strictly increase, $\phi^{-1}(x)$ is strictly increase.
    \item In terms of the symmetry, as $c(x)$ and $a_{\mathrm{p}} x$ are symmetry about the origin, $\phi^{-1}(x)$ is symmetry about the origin. 
    \item Intuitively, we have $\phi^{-1}(0) = 0$
\end{itemize}

Thus, $\phi^{-1}(x)$ is strictly increase, and strictly convex when $x\geq0$ and strictly concave when $x<0$, and $\phi^{-1}(x) +\phi^{-1}(-x) = 0, \phi^{-1}(0) = 0$.

Because the inverse function of a strictly increasing convex function is concave and strictly increased, and the inverse of a strictly increasing concave function is convex and strictly increased, we know $\phi(\lambda)$ strictly increases, strictly concave when $\lambda\geq 0$ and strictly convex when $\lambda < 0$, and $\phi(\lambda) + \phi(-\lambda) = 0, \phi(0)=0$. This completes the proof of Lemma \ref{aggressive_theorem_lamma2}.
\end{proof}

The analysis of Lemma \ref{aggressive_theorem_lamma2} implies $\phi(\lambda)$ is a piecewise convex and concave function with central symmetry property. We naturally provide the following Lemma to analyze the $t$th stage state-value function derivative in the optimality conditions from Lemma \ref{aggressive_theorem_lamma1} using the property of this $\phi(\lambda)$ function.

\begin{lemma}\emph{Property of reverse state-value function derivative.} \label{aggressive_theorem_lamma3}
    Consider the $\phi^{-1}(x)$ described in Lemma \ref{aggressive_theorem_lamma1} and its inverse function $\phi(\lambda)$, and take the price model described in Lemma \ref{aggressive_theorem_lamma1}, the following properties are true,
    \begin{enumerate}
        \item $a_{\mathrm{p}} \phi(\lambda) - \lambda$ strictly decreases with regards to $\lambda$;
        \item The $t$th stage reverse state-value function derivative from Lemma \ref{aggressive_theorem_lamma1} under price distribution $\gamma_{t+1}, \pi_{t+1}$ and expectation $\mathbb{E}[\lambda_{t+1}]$ satisfies,
        \begin{subequations}
        \begin{align}
            x_t &> 0, \label{property21} \\
            -q_t(x_t|\mathbb{E}[\lambda_{t+1}])/A &= \mathbb{E}[\lambda_{t+1}] -  a_{\mathrm{p}} Ax_t 
            +a_{\mathrm{p}} \phi(-\mathbb{E}[\lambda_{t+1}] + a_{\mathrm{p}} Ax_t ) > 0, \label{property22}\\
            -\mathbb{E}[q_t(x_t|\lambda_{t+1})]/A &= \mathbb{E}[\lambda_{t+1}] -  a_{\mathrm{p}} Ax_t 
            +a_{\mathrm{p}} [w_{\Lambda_{t+1}}\phi(-\pi_{t+1} + a_{\mathrm{p}} Ax_t )  + (1-w_{\Lambda_{t+1}})\phi(\gamma_{t+1} + a_{\mathrm{p}}A x_t)]  \nonumber\\
            &> -q_t(x_t|\mathbb{E}[\lambda_{t+1}])/A.\label{property23}
        \end{align}           
        \end{subequations}
    \end{enumerate}
\end{lemma}

\begin{proof}
    We first prove the \emph{property 1)}.

    According to Lemma \ref{aggressive_theorem_lamma1}, $\phi^{-1}(x) = c(x)+a_{\mathrm{p}} x$. Denote $\dot{}$ as the first-order derivative in terms of the function variable $x$ or $\lambda$, due to $\dot{c}(x)$ is positive, $\dot{\phi}^{-1}(x) > a_{\mathrm{p}}$, and according to the definition of inverse function, $\dot{\phi}(\lambda)< \frac{1}{a_{\mathrm{p}}}$.
    
    By taking derivative of the $a_{\mathrm{p}} \phi(\lambda) - \lambda$, 
    \begin{align}
        \frac{\partial (a_{\mathrm{p}} \phi(\lambda) - \lambda)}{\partial \lambda} = 
        a_{\mathrm{p}} \dot{\phi}(\lambda) -1 < 0 .
    \end{align}
    This means the function $a_{\mathrm{p}} \phi(\lambda) - \lambda$ strictly decreases with regrades to $\lambda$.

    Then, we prove the \emph{property 2)} by contradiction. For brevity, we omit the stage index $(t + 1)$ for variables and the subscript of $w_{\Lambda_{t+1}}$ when there is no ambiguity. 
    
    We first assume $x_t \leq 0$. Then, the $-\mathbb{E}[\lambda]+ a_{\mathrm{p}}Ax_t < 0$. Based on Lemma \ref{aggressive_theorem_lamma2}, we rewrite (\ref{property22}) as
    \begin{align}
        a_{\mathrm{p}} \phi(\mathbb{E}[\lambda]- a_{\mathrm{p}} Ax_t )
        < \mathbb{E}[\lambda] - a_{\mathrm{p}} Ax_t. \label{property2_case0}
    \end{align}
    According to the property 1), $a_{\mathrm{p}} \phi[\mathbb{E}[\lambda] - a_{\mathrm{p}} Ax_t] - [\mathbb{E}[\lambda] - a_{\mathrm{p}} Ax_t]$ strictly decreases with regrades to $\mathbb{E}[\lambda] - a_{\mathrm{p}} Ax_t$, and $\phi(0)=0$. Thus, with $\mathbb{E}[\lambda] - a_{\mathrm{p}}Ax_t > 0$, (\ref{property2_case0}) stands.
    
    Thus, we prove (\ref{property2_case0}) when assuming $x_t \leq 0$. Then we take this property into the optimality conditions from Remark \ref{remark_proof},
    \begin{align}
        a_{\mathrm{p}} p_1 + \sum_{\tau=1}^{t} c_{\tau}(x_{\tau}) = \mathbb{E}[\lambda] - a_{\mathrm{p}} Ax_t 
         + a_{\mathrm{p}} \phi(-\mathbb{E}[\lambda]  + a_{\mathrm{p}}A x_t). \label{optimality_condition}
    \end{align}
    
    As the reverse state-value function derivative (right-hand side) is positive, the left-hand side must be positive,
    \begin{align}
        a_{\mathrm{p}} p_1 + \sum_{\tau=1}^{t} c_{\tau}(x_{\tau}) > 0. \label{left_hand_side}
    \end{align} 
    
    Because $x_{\tau} = A x_{\tau-1} + p_{\tau}, \forall \tau \in [1,t]$, $x_0=0$, according to the objective of the demand model with quadratic action cost and super quadratic state cost, (\ref{left_hand_side}) means the state value $x_t$ at stage $t$ should be positive, and state and actions at prior stages should be non-negative. The zero actions and states come from the discount factor $A$. This contradicts our assumption of $x_t \leq 0$. Thus, we get $x_t > 0$. Also, following the same process, when $x_t > 0$, the (\ref{property22}) should still stand to satisfy the optimality conditions and $-\mathbb{E}[\lambda]+ a_{\mathrm{p}}Ax_t < 0$. 
    
    Then, we separate the two cases to analyze the (\ref{property23}). 

    \emph{Case (1):} $\gamma + a_{\mathrm{p}} Ax_t  < 0$. 
    Under this condition, after some moderately tedious algebra, we rewrite (\ref{property23}) as 
    \begin{subequations}
    \begin{align}
        w\phi(\pi - a_{\mathrm{p}} Ax_t) + (1-w)\phi(-\gamma - a_{\mathrm{p}} Ax_t ) < \phi(\mathbb{E}[\lambda] - a_{\mathrm{p}} Ax_t ), \label{property2_case1}  
    \end{align}
    where $\pi - a_{\mathrm{p}} Ax_t>0, -\gamma - a_{\mathrm{p}} Ax_t  >0, \mathbb{E}[\lambda] - a_{\mathrm{p}} Ax_t >0$, and $\phi(\lambda)$ is a concave when $\lambda \geq 0$. Then, according to Jensen's inequality of a strictly concave function~\cite{convex},
    \begin{align}
         w\phi(\pi - a_{\mathrm{p}} Ax_t) + (1-w)\phi(-\gamma - a_{\mathrm{p}} Ax_t)<  \phi[w\pi - wa_{\mathrm{p}} Ax_t 
        - (1-w)\gamma - (1-w)a_{\mathrm{p}} Ax_t]
        = \phi(\mathbb{E}[\lambda] - a_{\mathrm{p}} Ax_t),
    \end{align}    
    which prove (\ref{property2_case1}).
    \end{subequations}
    
    \emph{Case (2):} $\gamma+ a_{\mathrm{p}}A x_t 
    \geq 0$. 
    Under this condition, after some moderately tedious algebra, we rewrite (\ref{property23}) as 
    \begin{align}
        w\phi(\pi - a_{\mathrm{p}} Ax_t) - (1-w)\phi(\gamma + a_{\mathrm{p}}A x_t) < \phi(\mathbb{E}[\lambda] - a_{\mathrm{p}} Ax_t )    
    \end{align}   
    which is true according to the \emph{case (1)} analysis.
    
    Combining two cases shows that (\ref{property23}) is true. Taking this property into the optimality conditions from Lemma \ref{aggressive_theorem_lamma1}, we get a greater right-hand side of (\ref{optimality_condition}), a greater left-hand side follows, thus resulting in a positive and greater $x_t$.
    This completes the proof of Lemma \ref{aggressive_theorem_lamma3}.
\end{proof}

Lemma \ref{aggressive_theorem_lamma3} provides important properties of the $t$th stage state-value function derivative, helping us analyze the optimality conditions with price $\gamma_{t+1}, \pi_{t+1}$ and $\mathbb{E}[\lambda_{t+1}]$. The lack of properties with regards to more skewed price distribution $\Gamma$ motivates us to provide the following Lemma to analyze the sensitivity of $t$th stage state-value function derivative. 

\begin{lemma}\emph{Sensitivity of reverse state-value function derivative.} \label{aggressive_theorem_lamma4}
          Given the same conditions of Lemma \ref{aggressive_theorem_lamma3}, the following property is true:
          \begin{itemize}
              \item The $t$th stage reverse state-value function derivative from Lemma \ref{aggressive_theorem_lamma1}: $-v_t(x_t) = a_{\mathrm{p}} [w_{\Lambda_{t+1}}\phi(-\pi_{t+1} + a_{\mathrm{p}}A x_t) + (1-w_{\Lambda_{t+1}}) \phi(\gamma_{t+1} + a_{\mathrm{p}} Ax_t)] 
        + \mathbb{E}[\lambda_{t+1}] -  a_{\mathrm{p}} A x_t$ strictly increase with $\pi_{t+1}$.
          \end{itemize}      
\end{lemma}
 
 \begin{proof}   
    We prove this by showing the derivative of the reverse state-value function derivative $-v_t(x_t)$ with regards to $\pi_{t+1}$ is strictly positive.

        From Lemma \ref{aggressive_theorem_lamma3}, the following conditions are true (for brevity, we still omit the stage index $(t + 1)$ for variables and subscript of $w_{\Lambda_{t+1}}$ when there is no ambiguity.)
        \begin{subequations}
            \begin{align}
                x_t > 0, \\
                -\pi+ a_{\mathrm{p}} Ax_t < 0, \\
                \mathbb{E}[\lambda]- a_{\mathrm{p}}Ax_t > 0 .\label{Lemma4_derivative}
            \end{align}
        \end{subequations}
        
        According to (\ref{Lemma4_derivative}), 
        \begin{subequations}
        \begin{align}
            w\pi - (1-w) \gamma > a_{\mathrm{p}}Ax_t, \\
            w\pi - (1-w) \gamma > (1-w) a_{\mathrm{p}} Ax_t + w a_{\mathrm{p}}Ax_t, \\
            w(\pi - a_{\mathrm{p}} Ax_t) >(1-w) (\gamma + a_{\mathrm{p}}Ax_t).  \label{lemma4_condition2}
        \end{align}               
        \end{subequations}
        
        Due to $w \in [0,0.5]$, $\pi - a_{\mathrm{p}}Ax_t > \gamma + a_{\mathrm{p}}Ax_t$. Also, by definition, $\gamma+a_{\mathrm{p}}Ax_t > -\mathbb{E}[\lambda] + a_{\mathrm{p}}Ax_t$, and $\mathbb{E}[\lambda] < \pi$. Thus, let $\chi = -\pi+  a_{\mathrm{p}}Ax_t < 0$, $\alpha = -\gamma - a_{\mathrm{p}} Ax_t$, we get $\chi < \alpha < -\chi$.
        Also, from $w\pi - (1-w)\gamma = \mathbb{E}[\lambda] $, 
    \begin{align}
        w=\frac{\mathbb{E}[\lambda] + \gamma}{\pi+\gamma}, 1-w = \frac{\pi - \mathbb{E}[\lambda]}{\pi+\gamma},
    \end{align}
    Take $w,1-w$ inside and take the derivative of the function with regard to $\pi$, 
        \begin{align}
            &\frac{\partial [w\phi(\chi) +(1-w)\phi(-\alpha) + \mathbb{E}[\lambda] - a_{\mathrm{p}}A x_t]}{\partial \pi}  \nonumber \\
            &= -\frac{\mathbb{E}[\lambda] + \gamma}{(\pi+\gamma)^2} \phi(\chi) - 
            \frac{\mathbb{E}[\lambda] + \gamma}{\gamma+\pi} \phi'(\chi) + \frac{\mathbb{E}[\lambda] + \gamma}{(\pi+\gamma)^2} \phi(-\alpha)    
            = \frac{\mathbb{E}[\lambda] + \gamma}{(\gamma+\pi)^2} [\phi(-\alpha) - \phi(\chi)] - \frac{\mathbb{E}[\lambda] + \gamma}{\gamma+\pi} \phi'(\chi). \label{detivative_first_order}
            \end{align}

            According to the definition of the derivative, 
        \begin{align}
            \phi'(\chi) = \lim_{\alpha \rightarrow \chi} 
            \frac{\phi(\alpha) - \phi(\chi)}{\alpha - \chi}. 
        \end{align}
        Because $\chi<\alpha<-\chi$, and from Lemma \ref{aggressive_theorem_lamma2}, $\phi(\lambda)$ is symmetry about the origin, strictly increased, convex, and concave when $\lambda<0$ and $\lambda \geq0$, respectively, and $\phi(0)=0$, we conclude
        \begin{align}
            \phi'(\chi) \leq \frac{\phi(\chi) - \phi(\alpha)}{ \chi - \alpha} 
            < \frac{\phi(\chi)+\phi(\alpha)} 
            {\chi + \alpha} = \frac{-\phi(\alpha)-\phi(\chi)}{-\alpha-\chi}, \label{second_inequality}
        \end{align}
        and the second inequity holds because of the following,
        \begin{subequations}
            \begin{align}
               ( \phi(\chi)-\phi(\alpha)) (\chi+\alpha) 
               <  ( \phi(\chi)+\phi(\alpha)) (\chi-\alpha), \\
               \alpha \phi(\chi) < \chi \phi(\alpha), \\
               \phi(\chi) < \phi(\alpha) - \phi'(\chi)(\alpha -\chi), \\
               \alpha \phi(\alpha) - \alpha \phi'(\chi)(\alpha -\chi) - \chi \phi(\alpha) < 0, \\
               (\alpha -\chi) \phi(\alpha) < \alpha \phi'(\chi)(\alpha -\chi), \\
               \phi(\alpha) < \alpha \phi'(\chi), \\
               \frac{\phi(\alpha) - \phi(0)}{\alpha - 0} > \phi'(\chi).
            \end{align}
        \end{subequations}

        As $\phi(\lambda)$ symmetry about the origin, we know $-\phi(\alpha) = \phi(-\alpha)$ and 
            \begin{align}
            \phi'(\chi) < \frac{\phi(-\alpha)-\phi(\chi)} 
            {-\alpha- \chi} = \frac{\phi(-\alpha)-\phi(\chi)} 
            {\gamma+a_{\mathrm{p}} Ax_t- (-\pi+a_{\mathrm{p}} Ax_t )}  = \frac{\phi(-\alpha)-\phi(\chi)} 
            {\gamma+ \pi},
             \label{Lemma4_condition1}              
            \end{align}
            which means the sign of (\ref{detivative_first_order}) is positive.
            
        Thus, we show the first-order derivative of the reverse state-value function derivative $-v_t(x_t)$ with regards to $\pi$ (as (\ref{detivative_first_order}) shows) is positive. This implies that the reverse state-value function derivative strictly increases with $\pi$, and completes the proof of Lemma \ref{aggressive_theorem_lamma4}.
\end{proof}

Lemma \ref{aggressive_theorem_lamma4} indicates that with an increased price $\pi_{t+1}$, i.e., $\pi'_{t+1}$, the $t$th stage reverse state-value function derivative increases, i.e., the right-hand side in the optimality conditions described in (\ref{optimality_condition}). This implies the left-hand side also increases; thus, $x_t$ increases. By combining all these properties, we are able to complete the proof of Theorem \ref{aggressive_theorem}.

\begin{proof}\emph{Proof of Theorem \ref{aggressive_theorem}.}
    First, from Lemma \ref{aggressive_theorem_lamma1}, the optimality conditions of the demand model for the first stage action can be expressed as (\ref{optimality_conditions}). Note that we set $V_{t+1}(x_{t+1}) = 0$ as it doesn't affect the action and state before stage $t+1$. We rewrite part of it here for convenience
        \begin{align}
            a_{\mathrm{p}} p_1 + \sum_{\tau=1}^{t} c_{\tau}(x_{\tau}) =   \mathbb{E}[\lambda_{t+1}] - a_{\mathrm{p}} Ax_{t}  
     + a_{\mathrm{p}} [w_{\Lambda_{t+1}} \phi(-\pi_{t+1} + a_{\mathrm{p}} Ax_{t} ) + (1-w_{\Lambda_{t+1}}) \phi(\gamma_{t+1} + a_{\mathrm{p}} Ax_{t} )].  \label{rewrite}
        \end{align}
        
    According to Lemmas \ref{aggressive_theorem_lamma2}, \ref{aggressive_theorem_lamma3}, the right-hand side (reverse state-value function derivative) is positive, and greater than the optimality conditions under price expectation in Remark \ref{remark_proof}, indicating the positive and greater left-hand side in (\ref{rewrite}), so as to the $x_t$. Also, from Lemma \ref{aggressive_theorem_lamma4}, the right-hand side strictly increases with the $\pi_{t+1}$, and the left-hand side follows. Thus, we get the following property for the state-value function,
    \begin{align}
        \mathbb{E}_{\Gamma_{t+1}}[Q_t(x_{t}|\lambda_{t+1})]
        < \mathbb{E}_{\Lambda_{t+1}}[Q_t(x_{t}|\lambda_{t+1})] 
        < Q_t(x_{t}|\mathbb{E}_{\Lambda_{t+1}}[\lambda_{t+1}]).\label{key_state_for_theorem12} 
    \end{align}           
    

    By the model definition, $x_{\tau} = Ax_{\tau-1} + p_{\tau},\tau \in [1,t]$, $x_0=0$, the influence of $x_t$ degrades following the state transition with a discount factor of $A<1$. Also, the demand has a minimization objective with quadratic action cost and superquadratic state cost. Thus, all actions and states prior to the stage $t$ should be non-negative and non-decreasing with the price skewness. This means that price skewness causes higher or equal demand levels to change ahead of time, and this proves the Theorem. 
\end{proof}

\section{Proof of Corollary \ref{distribution}}
\begin{proof}
    The key to analyzing this Corollary is that the price distributions that satisfy the Corollary \ref{distribution} can always be discretized as a combination of different two-point price pairs as described in Theorem \ref{aggressive_theorem}.

        We first express the state-value function with price distribution $\Lambda$ as follows: (for brevity, we omit the stage index $t+1$ without ambiguity)
    \begin{subequations}
        \begin{align}
            \mathbb{E}_{\Lambda}[Q_t(x_t|\lambda)] = \int_{\lambda \in \mathcal{X}} f_{\Lambda}(\lambda) Q_t(x_t|\lambda) d\lambda,  \label{value_function_int}
        \end{align}
        where we separate the integral into two parts with price variable $\gamma$ and $\pi_{\Lambda}$, and get
        \begin{align}
           \mathbb{E}_{\Lambda}[Q_t(x_t|\lambda)] 
           =  \int_{-\infty}^{\mu} f_{\Lambda}(\gamma) Q_t(x_t|\gamma) d\gamma
           + \int_{\mu}^{+\infty} f_{\Lambda}(\pi_{\Lambda}) Q_t(x_t|\pi_{\Lambda}) d\pi_{\Lambda},  \label{11b}
        \end{align}
        and the same expression stands for the distribution $\Gamma$:
        \begin{align}
            \mathbb{E}_{\Gamma}[Q_t(x_t|\lambda)] 
            =  \int_{-\infty}^{\mu} f_{\Gamma}(\gamma) Q_t(x_t|\gamma) d\gamma
           + \int_{\mu}^{+\infty} f_{\Gamma}(\pi_{\Gamma}) Q_t(x_t|\pi_{\Gamma}) d\pi_{\Gamma}.  \label{11c}
        \end{align}
    \end{subequations}

    Then, we discretize the two-point price variables to many two-point price pairs (denoted with superscript $*$) with $w_{\Gamma},1-w_{\Gamma}$ and $w_{\Lambda},1-w_{\Lambda}$ probability, and satisfy expectation condition (\ref{complicated_pricemodel1}) $w_{\Gamma}\pi^*_{\Gamma} + (1-w_{\Gamma})\gamma^* = w_{\Lambda}\pi^*_{\Lambda} + (1-w_{\Lambda})\gamma^* = \mu$. From conditions (\ref{complicated_pricemodel2})-(\ref{complicated_pricemodel5}), $\pi^*_{\Gamma}>\pi^*_{\Lambda}$, $w_{\Gamma},w_{\Lambda} \in [0,0.5]$, aligning with the price model in Theorem \ref{aggressive_theorem}. Here we show the right skewness condition, and the left skewness case can be obtained as a mirror, as mentioned in Remark \ref{remark}.
    Then, each price pair should follow the same probability distribution as the variables $\gamma, \pi_{\Gamma}$ ($f_{\Gamma}$) and $\gamma,\pi_{\Lambda}$ ($f_{\Lambda}$), respectively, i.e., 
    \begin{align}
        \frac{f_{\Gamma}(\pi_{\Gamma}^*)}{f_{\Gamma}(\gamma^*)+f_{\Gamma}(\pi^*_{\Gamma})} 
        = w_{\Gamma}, 
        \frac{f_{\Lambda}(\pi_{\Lambda}^*)}{f_{\Lambda}(\gamma^*)+f_{\Lambda}(\pi^*_{\Lambda})} 
        = w_{\Lambda}, \forall \gamma^* \in \gamma, \pi_{\Gamma}^* \in \pi_{\Gamma}, \pi_{\Lambda}^* \in \pi_{\Lambda},\label{10}
    \end{align}
    and the expectation of all two-point price pairs is equivalent to the following with the PDF:
    \begin{align}
        \frac{f_{\Gamma}(\gamma^*) \gamma^*}{f_{\Gamma}(\pi_{\Gamma}^*) + f_{\Gamma}(\gamma^*)} 
        + \frac{f_{\Gamma}(\pi_{\Gamma}^*) \pi_{\Gamma}^*} {f_{\Gamma}(\pi_{\Gamma}^*) + f_{\Gamma}(\gamma^*)} 
        = \frac{f_{\Lambda}(\gamma^*) \gamma^*} {f_{\Lambda}(\pi_{\Lambda}^*) + f_{\Lambda}(\gamma^*)} 
        + \frac{f_{\Lambda}(\pi_{\Lambda}^*) \pi_{\Lambda}^*}{f_{\Lambda}(\pi_{\Lambda}^*)+ f_{\Lambda}(\gamma^*) } = \mu. \label{expectation}
    \end{align}

    According to our analysis of the discretization, all the two-point price pairs from $\Gamma,\Lambda$ distribution have the same expectations, $\pi^*_{\Gamma}>\pi_{\Lambda}^*$, and $w_{\Gamma},w_{\Lambda} \in [0,0.5]$, which satisfy the price model described in Theorem \ref{aggressive_theorem}. Thus, all price pairs satisfy:
    \begin{subequations}
    \begin{align}
        (1-w_{\Gamma})Q_t(x_t|\gamma^*) + w_{\Gamma}Q_t(x_t|\pi^*_{\Gamma}) < (1-w_{\Lambda})Q_t(x_t|\gamma^*) + w_{\Lambda}Q_t(x_t|\pi^*_{\Lambda}), \label{two_point}
    \end{align}
    taking (\ref{10}) into (\ref{two_point}), 
       \begin{align}
            \frac{Q_t(x_t|\gamma^*) f_{\Gamma}(\gamma^*) +
             f_{\Gamma}(\pi_{\Gamma}^*) Q_t(x_t|\pi^*_{\Gamma})}
             {f_{\Gamma}(\pi_{\Gamma}^*)+f_{\Gamma}(\gamma^*)}
            < \frac{Q_t(x_t|\gamma^*) f_{\Lambda}(\gamma^*)
            + f_{\Lambda}(\pi_{\Lambda}^*)
            Q_t(x_t|\pi^*_{\Lambda})}
            {f_{\Lambda}(\pi_{\Lambda}^*)+f_{\Lambda}(\gamma^*)}.
       \end{align}
    \end{subequations}

   Given the condition (\ref{seperate_condition}), $f_{\Gamma}(\pi_{\Gamma}^*)+f_{\Gamma}(\gamma^*) < f_{\Lambda}(\pi_{\Lambda}^*)+f_{\Lambda}(\gamma^*)$. Thus, for $\forall \gamma^* \in \gamma, \pi_{\Gamma}^* \in \pi_{\Gamma}, \pi_{\Lambda}^* \in \pi_{\Lambda}$,   
    \begin{align}
        f_{\Gamma}(\gamma^*)Q_t(x_t|\gamma^*) + f_{\Gamma}(\pi_{\Gamma}^*) Q_t(x_t|\pi^*_{\Gamma}) 
        <  f_{\Lambda}(\gamma^*)Q_t(x_t|\gamma^*) + f_{\Lambda}(\pi_{\Lambda}^*) Q_t(x_t|\pi^*_{\Lambda}).
        \label{inter_expand}
    \end{align}

    Now, by combining all price pairs,
    \begin{subequations}
    \begin{align}
         \int_{-\infty}^{+\infty} f_{\Gamma}(\gamma) Q_t(x_t|\gamma) d\gamma 
        + \int_{-\infty}^{+\infty} 
        f_{\Gamma}(\pi_{\Gamma}) Q_t(x_t|\pi_{\Gamma}) d\pi_{\Gamma} 
        < \int_{-\infty}^{+\infty} f_{\Lambda}(\gamma) Q_t(x_t|\gamma) d\gamma 
        + \int_{-\infty}^{+\infty} 
        f_{\Lambda}(\pi_{\Lambda}) Q_t(x_t|\pi_{\Lambda}) d\pi_{\Lambda}.
    \end{align}
    Due to $\gamma<\mu, \pi_{\Gamma} \geq \mu, \pi_{\Lambda}\geq \mu$, 
        \begin{align}
            \int_{-\infty}^{\mu} f_{\Gamma}(\gamma) Q_t(x_t|\gamma) d\gamma 
        + \int_{\mu}^{+\infty} 
        f_{\Gamma}(\pi_{\Gamma}) Q_t(x_t|\pi_{\Gamma}) d\pi_{\Gamma} 
        < \int_{-\infty}^{\mu} f_{\Lambda}(\gamma) Q_t(x_t|\gamma) d\gamma 
        + \int_{\mu}^{+\infty} 
        f_{\Lambda}(\pi_{\Lambda}) Q_t(x_t|\pi_{\Lambda}) d\pi_{\Lambda},
        \end{align}
        and from (\ref{11b}) and (\ref{11c}),
\begin{align}
    \mathbb{E}_{\Gamma}[Q_t(x_t|\lambda)] < \mathbb{E}_{\Lambda}[Q_t(x_t|\lambda)] 
\end{align}
    which shows given the price distribution conditions as (\ref{complicated_pricemodel}), demand shows prudence and proves this Corollary.
    \end{subequations}
\end{proof}


\end{document}